\documentclass[times,sort&compress,3p]{elsarticle}
\usepackage[labelfont=bf]{caption}

\usepackage{amsmath,amsfonts,amssymb,amsthm,booktabs,color,epsfig,graphicx,hyperref,url}

\usepackage{mathtools}

\theoremstyle{plain}
\newtheorem{theorem}{Theorem}

\newtheorem{proposition}{Proposition}

\newtheorem{corollary}{Corollary}

\theoremstyle{definition}
\newtheorem{definition}{Definition}
\newtheorem{remark}{Remark}
\newtheorem{example}{Example}

\begin{document}

\begin{frontmatter}

\title{Dimension reduction of multivariate densities in Bayes spaces} 

\author[olo]{Adéla Czolková\corref{mycorrespondingauthor}}
\author[olo]{Karel Hron}
\author[ber]{Sonja Greven}

\address[olo]{Univerzita Palackého v Olomouci, 17. listopadu 12, 77900 Olomouc, Česká republika}
\address[ber]{Humboldt-Universität zu Berlin, Unter den Linden 6, 10099 Berlin, Deutschland}

\cortext[mycorrespondingauthor]{Corresponding author. Email address: \url{adela.czolkova01@upol.cz}}

\begin{abstract}
The Bayes space provides a Hilbert space structure for analysing probability density functions (PDFs), equipping them with a geometry that reflects their relative and constrained nature. A key tool in this framework is the centred logratio (clr) transformation, which establishes an isometric isomorphism between the Bayes space and (a subspace of) the classical $L^2$ space. This makes it possible to apply functional data analysis (FDA) techniques, particularly functional principal component analysis (FPCA), to both univariate and multivariate density data in the context of dimension reduction. For multivariate PDFs, embedding them in the Bayes space enables an orthogonal decomposition into independent and interactive components. Furthermore, the independent part can be decomposed into mutually orthogonal geometric marginals. This structure provides more profound insights into the sources of variation in multivariate densities. We show that this decomposition of the total variance is  optimal in a PCA sense, impacting the interpretation of the eigenfunctions and scores resulting from FPCA. We demonstrate that applying FPCA directly to multivariate densities is equivalent in a certain sense to applying multivariate FPCA to their decomposed form, with the resulting eigenfunctions and scores decomposing accordingly. The unique decomposition based on these theoretical results is applied to housing and geological empirical data respectively, demonstrating the interpretability and practical value of this approach. 
\end{abstract}


\end{frontmatter}

\newcommand{\B}{\mathcal{B}}
\newcommand{\Lz}{L_0^2}
\newcommand{\clr}{\mathrm{clr}}
\newcommand{\intr}{\mathrm{int}}
\newcommand{\ind}{\mathrm{ind}}
\newcommand{\X}{\mathrm{X}}
\newcommand{\Y}{\mathrm{Y}}
\newcommand{\bs}{\boldsymbol}
\newcommand{\dd}{\mathrm{d}}

\section{Introduction} 
\label{sec:intro}

The analysis of PDFs as data objects is an established topic in functional data analysis, and several approaches for their dimension reduction are discussed in the literature \cite{kneip01,petersen22}. The Bayes space methodology for the representation of PDFs as scale-invariant objects \cite{egozcue06,boogaart14} has turned out to provide a strong methodological framework with implications for functional data analysis and beyond, and the field is constantly growing \cite{delicado11,hron16,menafoglio14,menafoglio16,menafoglio18,talska18,lei23,eckardt24,ma24,murph24,qiu24,kutta25,maier25}. For the case of multivariate PDFs, \cite{genest23} developed an orthogonal decomposition, which enables to identify an independent part (contained essentially in reformulated univariate marginals) as well as interactions among variables. This decomposition has many important properties, which do not have a counterpart in standard probability theory, including a generalized Pythagoras' theorem for (squared) norms and a marginal-free property for PDFs from the interaction part.

In this paper, we introduce two important implications of this orthogonal decomposition in the context of dimension reduction, which result in properties that do not have a counterpart in the classical $L^2$ setting. First, we demonstrate that this decomposition is optimal in the sense of a principal component analysis. Following preliminary considerations in \cite{boogaart22}, this is also reflected in the decomposition of the total variance of the PDF sample. This finding allows us to perform effective diagnostics and determine which parts of the decomposition are essential for the multivariate PDF structure. In particular, this can be used to evaluate the importance of individual (geometric) marginals and to identify the specific interactions that dominate the density data structure. Secondly, building on the initial developments in \cite{hron16} and \cite{filzmoser21}, we turn to functional principal component analysis itself. We derive a one-to-one correspondence between the loadings and the scores of multivariate functional principal component analysis (mFPCA) for the vectors of multivariate PDFs from the orthogonal decomposition and those obtained from a functional principal component analysis (FPCA) of the multivariate densities. This can be used to perform diagnostics on the FPCA loadings and scores, specifically to reveal the contribution of each component of the orthogonal decomposition (i.e., geometric marginals and specific interactions) to the overall picture provided by the FPCA results for multivariate PDFs.

The paper is organized as follows. In Section~\ref{sec:bs}, the basics of Bayes spaces are summarized, and in Section~\ref{sec:decomp}, the orthogonal decomposition of multivariate PDFs is recalled. Section~\ref{sec:pca} develops the principal component analysis perspective of the decomposition, and in Section~\ref{sec:fpcadec}, the implications of the decomposition for loadings and scores of both FPCA for multivariate densities and mFPCA are derived. Section~\ref{sec:app} contains applications to housing and geological data, and Section~\ref{sec:conc} concludes.


\section{Bayes spaces}
\label{sec:bs}

Bayes spaces provide a geometric representation for probability density functions (PDFs), which possess the property of scale invariance. More specifically, two positive functions (densities) $f$ and $g$ on the same domain $\Omega$ carry the same relative information \cite{genest23,hron23}, if they are proportional $g(x) = c f(x)$ for all $x \in \Omega$ for some positive real number $c$. These densities are said to be equivalent, denoted $f=_\B g$. This is a direct consequence of the scale invariance of the associated measures $\mu$ and $\nu$ such that their densities are $f = \dd\mu/\dd\lambda$ and $g = \dd\nu/\dd \lambda$
for a finite positive real-valued reference measure $\lambda$ on a measurable space $(\Omega, \mathcal{A})$.

Given such a reference measure $\lambda$, the Bayes space $\B(\lambda)$ can be defined \cite{boogaart14}.
\begin{definition}
    The Bayes space $\B(\lambda)$ is a quotient space, the space of equivalence classes of ($\lambda$ almost everywhere) positive densities on $(\Omega, \mathcal{A})$, with respect to the relation $=_\B$.
\end{definition}

\noindent In this space, we can define two operations -- perturbation $\oplus$ and powering $\odot$ -- as counterparts to the usual sum of two real functions and the multiplication of a function by a scalar in the $L^2$ space, giving $\B(\lambda)$ a vector space structure.

\begin{definition}
    Given two arbitrary densities $f,g \in \B(\lambda)$ and a number $\alpha \in \mathbb{R}$,
    \begin{equation*}
        (f \oplus g) =_\B f \cdot g, \quad (\alpha \odot f) =_\B f^{\alpha}.
    \end{equation*}
\end{definition}

\noindent Suppose a finite reference measure $\lambda$ and define a subspace $\B^2(\lambda)$ of space $\B(\lambda)$ as those elements that are square-log integrable densities with respect to $\lambda$, i.e.,
\begin{equation}\label{eq:b2}
    \B^2(\lambda) = \left\lbrace f \in \B(\lambda) \ \Big| \, \int_\Omega \left| \ln (f) \right| ^2 \dd \lambda < +\infty \right\rbrace.
\end{equation}
Considering this subspace, we can define the inner product of two densities in $\B^2(\lambda)$.
\begin{definition}
    For any $f,g \in \B^2(\lambda)$, the inner product is
    \begin{equation*}
        \left\langle f,g \right\rangle_{\B^2(\lambda)} = \frac{1}{2 \lambda(\Omega)} \int_{\Omega} \int_{\Omega} \ln \frac{f(s)}{f(t)} \, \ln \frac{g(s)}{g(t)} \, \dd\lambda(s) \, \dd\lambda(t).
    \end{equation*}
\end{definition}

\noindent With this inner product, we obtain a Hilbert space structure of $\B^2(\lambda)$. The inner product also implies the norm and the distance
\begin{equation*}
    \| f \|_{\B^2(\lambda)} = \sqrt{\left\langle f,f \right\rangle_{\B^2(\lambda)}}, \qquad \dd_{\B^2(\lambda)}(f,g) = \| f \ominus g \|_{\B^2(\lambda)},
\end{equation*}
where $f \ominus g = f \oplus [(-1)\odot g]$ is a perturbation-subtraction.

The densities are not usually processed directly in the Bayes space, but they are mapped into a subspace of the $L^2(\lambda)$ space (the space of real-valued square-integrable functions), where the standard methods of functional data analysis can be used \cite{boogaart14}. For this purpose, the centered logratio transformation, $\clr: \B^2(\lambda) \rightarrow L^2(\lambda)$, was defined as
\begin{equation}\label{eq:clr}
    \clr(f) \equiv f_c = \ln (f) - \frac{1}{\lambda(\Omega)} \int_{\Omega} \ln (f) \, \dd \lambda, \quad f \in \B^2(\lambda).
\end{equation}
For any $f\in\B^2(\lambda)$, the transformation returns a function $\clr(f)\! : \Omega \rightarrow \mathbb{R}$.
Moreover, the zero-integral constraint,
\begin{equation*}
    \int_{\Omega} \clr(f) \, \dd\lambda = 0,
\end{equation*}
holds, such that $\clr(f)$ lies in a subspace of $L^2(\lambda)$ of zero-integral functions, denoted by $\Lz(\lambda)$ and referred to as clr space in the context of Bayes spaces. This transformation, from $\B^2(\lambda)$ to $\Lz(\lambda)$, is an isometric isomorphism (in particular, it is a linear mapping). The inverse clr transformation, $\clr^{-1}$, for a function $f_c\in \Lz(\lambda)$ is defined as $\clr^{-1}(f_c) =_\B\exp \lbrace f_c \rbrace \in \B^2(\lambda)$. 

For a more detailed introduction to Bayes spaces, see, for example, \cite{hron16} and \cite{talska20}. We here focus on the case where densities are fully observed, as can typically approximately be assumed in large datasets such as in our applications. While not the focus here, in practice  the following issues may need to be considered. First, the case of small samples per density requires special care. \cite{pavluu2024principal} discuss correction approaches for FPCA for distributions observed by samples in Bayes spaces. We will also revisit this case in the discussion. Second, the necessity for all PDFs in the sample to have identical support can prohibit the use of Bayes spaces in some cases where positive density values on the considered domain cannot be reasonably assumed. Third, we typically assume as in FDA that the support is bounded. This is not a hard restriction, however, as for unbounded supports an alternative (finite) reference measure to the standard Lebesgue measure can also be considered;  the standard normal reference measure is a common choice \cite{boogaart14}. The possibility of choosing an alternative reference measure also has the advantage of enabling  to downweight parts of the domain with insufficient sampling, or with a higher sampling error from the measurement device \cite{talska20}. Despite these potential limitations, the Bayes space approach offers several advantages compared to the $L^2$ space or its competitors. These include a Hilbert space structure that respects the constrained nature of densities, subcompositional coherence \cite{maier25} ensuring consistency of results on subdomains, and in particular extendability to  multivariate densities \cite{petersen22} as relevant here. The properties present in the context of dimension reduction with FPCA are elaborated on in the next sections.


\section{Multivariate densities and their orthogonal decomposition} \label{sec:decomp}

We  now introduce Bayes spaces  for multivariate densities as well as their orthogonal decomposition into (geometric) marginals and interactive parts, which we build on for their FPCA. In the case of $d$-variate densities, let $(\Omega,\mathcal{A})$ be a $d$-dimensional product measure space with a finite product reference measure $\lambda$, where
\begin{equation*}
    \Omega = \Omega_1 \times \cdots \times \Omega_d, \quad \mathcal{A} = \mathcal{A}_1 \otimes \cdots \otimes \mathcal{A}_d, \quad \lambda = \lambda_1 \otimes \cdots \otimes \lambda_d.
\end{equation*}
Then the Bayes spaces $\B(\lambda)$, $\B(\lambda_1)$, \dots, $\B(\lambda_d)$, and their subspaces $\B^2(\lambda)$, $\B^2(\lambda_1)$, \dots, $\B^2(\lambda_d)$ can be defined in a similar manner as stated above in (\ref{eq:b2}), likewise for the clr transformation (\ref{eq:clr}) and the clr spaces $L^2_0(\lambda)$, $L^2_0(\lambda_1)$, \dots, $L^2_0(\lambda_d)$ \cite{genest23}.

Moreover, let $I$ be an arbitrary non-empty subset of $D = \lbrace 1,\dots,d \rbrace$, i.e., $\emptyset \neq I \subseteq D$, and let
\begin{equation*}
    \Omega_I = \bigtimes_{i \in I} \Omega_i, \quad \mathcal{A}_I = \bigotimes_{i \in I} \mathcal{A}_i, \quad \lambda_I = \bigotimes_{i \in I} \lambda_i,
\end{equation*}
then $\B(\lambda_I)$ is the Bayes space on $(\Omega_I,\mathcal{A}_I)$, $\B^2(\lambda_I)$ is its subspace of square-log integrable densities, and $\Lz(\lambda_I)$ the respective clr space.

Further, for any non-empty $I \subseteq D$, let $\B^2_I(\lambda)$ be a subspace of $\B^2(\lambda)$ such that
\begin{equation*}
    \B^2_I(\lambda) = \left\lbrace f \in \B^2(\lambda) \, \big| \, f =_\B f^I \cdot 1, \ f^I \in \B^2(\lambda_I), \ 1 \in \B^2(\lambda_{D\backslash I}) \right\rbrace,
\end{equation*}
where $1=\dd\lambda_{D\backslash I}/\dd\lambda_{D\backslash I}$ is a neutral element in $\B^2(\lambda_{D\backslash I})$ -- the uniform density. If $I=D$, it holds that $\B^2_D(\lambda) = \B^2(\lambda)$; and $\B^2_\emptyset(\lambda)$ is a trivial subspace of $\B^2(\lambda)$ containing only a $d$-variate uniform density $\dd\lambda/\dd\lambda$. Similarly as above, there exists the respective clr space $L^2_{0,I}(\lambda)$.

To keep in line with the geometry of Bayes spaces, the marginal information on one of the variables -- instead of by standard marginal densities (called arithmetic marginals here) -- is captured by the geometric marginals, which are orthogonal projections of a density $f$ onto $\B^2_I(\lambda)$ as shown in \cite{genest23}.

\begin{proposition}
    Given an arbitrary density $f \in \B^2(\lambda)$ and a non-empty set $I \subset D$, the unique orthogonal projection of $f$ onto $\B^2_I(\lambda)$ is given by
    \begin{equation*}
        f^I = \exp \left\lbrace \frac{1}{\lambda_{D\backslash I}(\Omega_{D\backslash I})} \int_{\Omega_{D\backslash I}} \ln(f) \ \dd\lambda_{D\backslash I} \right\rbrace,
    \end{equation*}
    and then
    \begin{equation*}
        \clr(f^I) = \frac{1}{\lambda_{D\backslash I}(\Omega_{D\backslash I})} \int_{\Omega_{D\backslash I}} \clr(f) \ \dd\lambda_{D\backslash I} \ \in L^2_{0,I}(\lambda).
    \end{equation*}
\end{proposition}

\begin{definition}
    For any non-empty $I \subset D$, $f^I$ is called the $I$th geometric marginal of $f$. When $I = \lbrace i \rbrace, \, i \in D$, the geometric marginal is denoted $f^i$ instead of $f^I$ or $f^{\lbrace i \rbrace}$, and $\B^2_i(\lambda) = \B^2_{\lbrace i \rbrace}(\lambda)$. Also, $f^D = f \in \B^2(\lambda)$ and
    \begin{equation*}
        f^\emptyset = \exp \left\lbrace \frac{1}{\lambda(\Omega)} \int_\Omega \ln (f) \ \dd\lambda \right\rbrace =_\B 1 = \frac{\dd\lambda}{\dd\lambda}.
    \end{equation*}
\end{definition}

\noindent A key property of multivariate PDFs in Bayes spaces is their orthogonal decomposition, derived in \cite{genest23}.

\begin{proposition}
\label{prop:decomp}
Any density $f \in \B^2(\lambda)$ can then be decomposed into independent and interactive part as follows,
\begin{equation} \label{eq:decomp1}
    f = f^\ind \oplus f^\intr = \bigoplus_{i=1}^d f^i \oplus \bigoplus_{I\subseteq D, |I| \geq 2} f^{I,\intr},
\end{equation}
where
\begin{equation} \label{eq:decomp2}
    f^{I,\intr} = \bigoplus_{J\subseteq I,J\neq\emptyset} \lbrace (-1)^{|I\backslash J|} \rbrace \odot f^J = f^I \ominus \left( \bigoplus_{J\subset I,|J|\geq2} f^{J,\intr} \oplus \bigoplus_{i \in I} f^i \right).
\end{equation}
\end{proposition}

\noindent Altogether, each $d$-dimensional density can be decomposed into $2^d-1$ parts, where the independent part consists of $d$ "univariate" geometric marginals, and the interactive part is composed of $2^d-d-1$ interactions among all possible combinations of variables ($I\subseteq D$). While geometric marginals capture the univariate information from the multivariate PDF, the interactive part enables to clearly separate single interaction effects. Note that in the truly independence case, the interactive part equals the uniform density, and both arithmetic and geometric marginals coincide. 

\begin{example}
    Let $d=2$, then $\Omega = \Omega_1 \times \Omega_2$, $\lambda = \lambda_1 \otimes \lambda_2$, and $D = \lbrace1,2 \rbrace$. Geometric marginals are defined as
    \begin{eqnarray*}
        f^1 = \exp \left\lbrace \frac{1}{\lambda_2(\Omega_2)} \int_{\Omega_2} \ln(f) \ \dd\lambda_2 \right\rbrace \quad \text{and} \quad f^2 = \exp \left\lbrace \frac{1}{\lambda_1(\Omega_1)} \int_{\Omega_1} \ln(f) \ \dd\lambda_1 \right\rbrace,
    \end{eqnarray*}
    the independent part equals $f^\ind = f^1 \oplus f^2$ and for the interactive part it holds that $f^\intr = f \ominus f^1 \ominus f^2$.
\end{example}

All parts of the decomposition are orthogonal projections onto the respective subspaces of $\B^2(\lambda)$, in particular, for any $i \in D$, $f^i$ is an orthogonal projection of $f$ onto $\B^2_i(\lambda)$, and for any $I\subseteq D, |I|\geq2$, $f^{I,\intr}$ is an orthogonal projection of $f$ onto the interaction subspace
\begin{equation*}
    \B^2_{I,\intr} = \left\lbrace g \in \B^2_I(\lambda) \, \big| \, \forall h \in \B^2_{I,\oplus}(\lambda): \langle g,h \rangle_{\B^2(\lambda)} = 0 \right\rbrace,
\end{equation*}
where
\begin{equation*}
    \B^2_{I,\oplus}(\lambda) = \bigoplus_{i\in I} \B^2_i(\lambda) \oplus \bigoplus_{J\subset I,|J|\geq2} \B^2_{J,\intr}(\lambda).
\end{equation*} 
It can be shown that
\begin{equation*}
    \B^2(\lambda) = \bigoplus_{i \in D} \B^2_i(\lambda) \oplus \bigoplus_{I\subseteq D,|I|\geq2} \B^2_{I,\intr} (\lambda)
\end{equation*}
and all subspaces in the direct sum are mutually orthogonal, so the density $f$ is decomposed into $2^d-1$ orthogonal parts and this decomposition is unique \cite{genest23}.

While subspaces $\B^2_{I,\intr}$ are called interaction spaces, the space
\begin{equation*}
    \B^2_\ind(\lambda) = \bigoplus_{i\in D} \B^2_i(\lambda) = \left\lbrace f \in \B^2(\lambda) \, \big| \, \exists f^i \in \B^2_i(\lambda),i\in D : f = f^1 \oplus\cdots\oplus f^d \right\rbrace
\end{equation*}
is called the independence space, and $f^\ind$ is an orthogonal projection of $f$ onto $\B^2_\ind(\lambda)$.

The orthogonal decomposition of $f \in \B^2(\lambda)$ from Proposition \ref{prop:decomp} can be formulated in the clr space $\Lz(\lambda)$ as well, then
\begin{equation} \label{eq:dec1.clr}
    \clr(f) = \clr(f^\ind) + \clr(f^\intr) = \sum_{i \in D} \clr(f^i) + \sum_{\substack{I\subseteq D, 
    |I|\geq2}} \clr(f^{I,\intr}),
\end{equation}
where
\begin{equation} \label{eq:dec2.clr}
    \clr(f^{I,\intr}) = \sum_{\substack{J\subseteq I, 
    J\neq\emptyset}} (-1)^{|I\backslash J|} \, \clr(f^J) = \clr(f^I) - \sum_{\substack{J\subset I, |J|\geq2}} \clr(f^{J,\intr}) - \sum_{i \in I} \clr(f^i)
\end{equation}
for $I\subseteq D,|I|\geq2$. All parts are elements of the corresponding orthogonal subspaces of $\Lz(\lambda)$, namely $L^2_{0,i}(\lambda), i\in D$, and $L^2_{0,I,\intr}(\lambda), I\subseteq D, |I|\geq2$.

In the following, we will focus on the case where $\lambda$ is the Lebesgue reference measure, and leave out $\lambda$ from the notation for simplicity in the following sections. Nevertheless, the results also hold for other reference measures $\lambda$.


\section{Principal component analysis perspective of the decomposition in multivariate Bayes spaces}
\label{sec:pca}

In Bayes spaces, each PDF is considered as one data object (functional variable). It is thus natural to extend this setting to the $P$-variate case, when a vector of PDFs is considered -- to what is commonly called multivariate functional data in FDA \cite{ramsay05}. The goal of this section is to perform a principal component analysis (PCA) on objects from the orthogonal decomposition of PDFs (where $P=2^d - 1$), as a functional counterpart to the approach for compositional data introduced in \cite{wang15}. However, we first introduce the general setting of multivariate density-valued data. Given a sample of PDFs defined on a compact domain $\Omega = \Omega_1 \times \cdots \times \Omega_d \subset \mathbb{R}^d$, $d \in \mathbb{N}$, this setting leads to a $N\times P$ "data matrix" of PDFs $\mathbf{F}=(f_n^p)$, $f_n^p \in \B^2$ on $\Omega$ with columns $\mathbf{f}^p$, for observations $n\in\{1,\ldots,N\}$, and densities $p\in\{1,\ldots,P\}$. Moreover, if all PDFs in each row of $\mathbf{F}$ have the same domain (which is, e.g., the case for elements of the orthogonal decomposition from \cite{genest23}), their covariance structure can be constructed as follows.  


\subsection{Covariance structure of PDFs}

To construct the covariance structure of PDFs, we first consider the (scalar) sample variance of PDFs, extending the compositional case in \cite{wang15}, as 
\begin{equation*}
    \mathrm{var}_{\B^2}(\mathbf{f}^p) = \frac{1}{N} \sum_{n=1}^N \| f_n^p \ominus \overline{f^p} \|_{\B^2}^2,
\end{equation*}
where $\mathbf{f}^p$ is the $p$th column of $\mathbf{F}$ and $\overline{f^p} = (1/N) \odot \bigoplus_{n=1}^N f_n^p$, $p\in\{1,\dots,P\}$. It can be considered as a generalization of the concept of total variance \cite{wang15}, known from compositional data analysis \cite{pawlowsky15} as well as from functional data analysis \cite[Eq. (11.8)]{kokoszka17}. As we assume the same domain for each PDF in each sample and component, we can proceed to define the sample covariance between two PDF variables as
\begin{equation}
\label{eq:cov}
\mathrm{cov}_{\B^2}(\mathbf{f}^p,\mathbf{f}^r) = \frac{1}{N} \sum_{n=1}^N \langle f_n^p \ominus \overline{f^p}, f_n^r \ominus \overline{f^r} \rangle_{\B^2}.
\end{equation}
It can be readily shown that $\mathrm{cov}_{\B^2}(\mathbf{f}^p, \mathbf{f}^p) = \mathrm{var}_{\B^2}(\mathbf{f}^p)$.
\begin{definition}
    Given an arbitrary $N\times P$ data matrix $\mathbf{F}$ of densities  $f_n^p \in \B^2$, the sample Bayes covariance matrix is the $P \times P$ matrix
    \begin{equation} \label{eq:covmat}
        \mathbf{C}_{\B^2} \equiv \mathbf{C}_{\B^2}(\mathbf{F}) = (\mathrm{cov}_{\B^2}(\mathbf{f}^p,\mathbf{f}^r))_{p,r \in \{1,\dots,P\}}.
    \end{equation}
\end{definition}
\noindent This covariance matrix $\mathbf{C}_{\B^2}(\mathbf{F})$ is positive semi-definite, which follows from its construction.

One particular case, being of interest here, is when the $n$-th row of $\mathbf{F}$ is 
\begin{equation*}
    \mathbf{f}_n = \left( f_n^1,\dots,f_n^d,f_n^{I_1,\intr},\dots,f_n^{I_{P-d},\intr} \right)^\top \in \B^2_1 \times \cdots \times \B^2_d \times \B^2_{I_1,\intr} \times \cdots \times \B^2_{I_{P-d},\intr}, 
\end{equation*}
i.e., where the row elements come from the orthogonal decomposition of a $d$-variate PDF $f_n$ \cite{genest23} with $I_1,\dots,I_{P-d}$ indexing the subsets $I \subseteq \lbrace 1,\dots,d \rbrace,\, |I|\geq2$, and thus $P=2^d-1$. To simplify the notation, let $\mathcal{I} = \lbrace 1,\ldots,d,(I_1,\intr),\ldots,\\(I_{P-d},\intr) \rbrace$ denote a set of indices of parts of the orthogonal decomposition. Then $\mathbf{f}_n = (f_n^i)_{i \in \mathcal{I}} \in \bigtimes_{i\in\mathcal{I}}\B^2_i$. The fact that the elements of $\mathbf{f}_n,\,n\in\{1,\dots,N\}$, are mutually orthogonal has an important consequence for the covariance structure.

\begin{theorem}
\label{t:covmat}
    Let, for all $n \in \{1,\dots,N\}$, elements of $\mathbf{f}_n$ come from the orthogonal decomposition of a $d$-variate density $f_n$. Then the covariance structure is diagonal, i.e.,
    \begin{eqnarray*}
        \mathbf{C}_{\B^2}(\mathbf{F}) = \mathrm{diag} \left\lbrace \mathrm{var}_{\B^2}(\mathbf{f}^i),\ i\in\mathcal{I} \right\rbrace.
    \end{eqnarray*}
\end{theorem}

\begin{proof}[\textbf{\upshape Proof:}]
    Suppose that without loss of generality each density $f_n$ is centered (otherwise, we can always center them, i.e., $\tilde{f_n} = f_n \ominus \overline{f}, \ n \in \{1,\dots,N\}, \ \overline{f} =  (1/N) \odot \bigoplus_{n=1}^N f_n$), then, thanks to orthogonality of the decomposition, it is obvious that
    \begin{equation*}
        \mathrm{cov}_{\B^2}(\mathbf{f}^p,\mathbf{f}^r) = \frac{1}{N} \sum_{n=1}^N \langle f_{n}^p, f_{n}^r \rangle_{\B^2} = 0,
    \end{equation*}
    where $p,r \in \mathcal{I}, \, p \neq r$. Accordingly, the covariance structure in this particular case is diagonal.
\end{proof}

\noindent The resulting covariance matrix will be the same for PDFs in the Bayes space as well as in the clr space, because the respective inner products (and consequently the norms) in these two spaces are equal due to the isometry of the clr transformation. 


\subsection{PCA for decomposed densities}

The covariance structure of the decomposed densities, represented by the sample Bayes covariance matrix, is key to performing PCA and deriving the main modes of variability. As all PDFs from the decomposition are guaranteed to have the same domain, the $j$th principal component is defined as a linear combination
\begin{equation*}
    \mathbf{v}^j = \bigoplus_{p=1}^P (\gamma_{jp} \odot \mathbf{f}^p) = (\gamma_{j1} \odot \mathbf{f}^1) \oplus \cdots \oplus (\gamma_{jP} \odot \mathbf{f}^P),
\end{equation*}
where the vector $\bs{\gamma}_j = (\gamma_{j1},\dots,\gamma_{jP})^\top$ maximizes $\mathrm{var}_{\B^2}(\mathbf{v}^j)$ subject to $\|\bs{\gamma}_j\|=1$ and $\bs{\gamma}_j^\top\bs{\gamma}_k=0$, for $k<j, \ j \in \{1, \ldots, P\}$. The vector $\bs{\gamma}_j$ can be obtained as $j$th eigenvector from the eigendecomposition of $\mathbf{C}_{\B^2}(\mathbf{F}) = \bs{\Gamma}\bs{\Lambda}\bs{\Gamma}^\top$, where $\mathbf{\Gamma} = (\bs{\gamma}_1, \dots, \bs{\gamma}_P)$ and the diagonal matrix $\bs{\Lambda}$ contains the respective eigenvalues.

If the covariance structure is diagonal, then the eigenvalues of $\mathbf{C}_{\B^2}(\mathbf{F})$ are directly equal to the variances of the original variables, and $ \bs{\Gamma}$ equals to the identity matrix. Accordingly, for the orthogonal decomposition of a sample of $d$-variate PDFs, the principal components directly equal elements of the decomposition.
\begin{proposition}
    Consider a sample of decomposed $d$-variate densities. Let $\pi : \{1, \dots, P\} \rightarrow \mathcal{I}$ be bijective
    such that
    \begin{equation*}
        \mathrm{var}_{\B^2}(\mathbf{f}^{\pi(1)}) \geq \dots \geq \mathrm{var}_{\B^2}(\mathbf{f}^{\pi(P)}),
    \end{equation*}
    then
    \begin{equation*}
        \mathbf{v}^1 = \mathbf{f}^{\pi(1)}, \ \dots, \mathbf{v}^{P} = \mathbf{f}^{\pi(P)}.
    \end{equation*} 
\end{proposition}
\begin{proof}[\textbf{\upshape Proof:}]
    In the case of the sample of decomposed densities, the covariance structure reflects the orthogonality of the decomposition, i.e., $\mathbf{C}_{\B^2}(\mathbf{F})$ is diagonal as stated in Theorem \ref{t:covmat}. This means that $\mathbf{C}_{\B^2}(\mathbf{F}) = \bs{\Lambda} = \mathrm{diag}(\lambda_1,\dots,\lambda_P)$ and $\bs{\Gamma}$ is the identity matrix. Then, according to the assumption of the proposition, $\lambda_1 = \mathrm{var}_{\B^2}(\mathbf{v}^1) =   \mathrm{var}_{\B^2}(\mathbf{f}^{\pi(1)})$ and consequently $\mathbf{v}^1 = \mathbf{f}^{\pi(1)}$. Similarly, for other eigenvalues, variances, and components.
\end{proof}
In other words, the orthogonal decomposition of \cite{genest23} is optimal in the sense of a principal component analysis, and the order of variances of the $\mathbf{f}^{i}$, $i \in \mathcal{I}$, corresponds to the order of principal components.

This finding has a direct implication in the possibility for dimension reduction, neglecting such elements of the orthogonal decomposition which contribute only marginally to the total variance of the sample of PDFs. This can be assessed using any diagnostic tool known from principal component analysis, e.g., the scree plot \cite{johnson07}. Another direct implication, already outlined in \cite{boogaart22}, but seen here as a consequence of the optimality property above, is the decomposition of the sample total variance of the original densities $f_n$, which is stated in the following corollary.

\begin{corollary}
    Let $\mathrm{var}_{\B^2}(\mathbf{f})$ denote the sample total variance of the PDFs $f_n\in\B^2$ and $\mathrm{var}_{\B^2}(\mathbf{f}^i), i\in\mathcal{I}$ the sample total variances of their orthogonal decomposition. Then 
    \begin{equation*}
        \mathrm{var}_{\B^2}(\mathbf{f})=\sum_{i\in\mathcal{I}} \mathrm{var}_{\B^2}(\mathbf{f}^i).
    \end{equation*}
\end{corollary}

\noindent This relation is analogous to Pythagoras' Theorem for the squared norms of PDFs from an orthogonal decomposition, as given in \cite{genest23}.


\section{Decomposition of FPCA}
\label{sec:fpcadec}

One of the main goals of unsupervised learning is dimension reduction. In the previous section, we demonstrated that FPCA nicely reflects the orthogonal decomposition of PDFs at the object level, resulting in a decomposition of their total variance. In this section, we go a step further and analyse the relationship between the eigenfunctions (loadings) and the scores of multivariate FPCA (mFPCA) for the elements of the orthogonal decomposition and the decomposition of the eigenfunctions and the scores from FPCA for multivariate PDFs. Specifically, we demonstrate that the loadings obtained from mFPCA are precisely the same as the orthogonal decompositions of the loadings obtained from FPCA for multivariate PDFs. Furthermore, the scores can also be decomposed according to this orthogonal decomposition. This demonstrates that the decomposition of eigenfunctions and scores is unique. This is key to understanding how the orthogonal decomposition of PDFs determines the structure of loadings and eigenfunctions, and to using these for dimension reduction and data exploration in practice, as illustrated in Section~\ref{sec:app}.

To gain a better understanding, we first formulate FPCA for multivariate densities.


\subsection{FPCA for multivariate densities}
\label{sec:fpca}

In this section, we extend the simplicial functional principal component analysis \cite{hron16}, how FPCA for univariate densities was called, to the multivariate case.

Let $f_1, \dots, f_N$ be a sample of multivariate probability density functions (PDFs) defined on a compact domain $\Omega = \Omega_1 \times \cdots \times \Omega_d \subset \mathbb{R}^d$. We consider these PDFs as elements of the Bayes space $\B^2$ on $\Omega$, and suppose these densities to be centered; if they are not, we can always center them, i.e., $\tilde{f_n} = f_n \ominus \overline{f}, \ n \in \{1,\dots,N\}, \ \text{where} \ \overline{f} =  (1/N) \odot \bigoplus_{n=1}^N f_n$.

We aim to find the main modes of variability of the PDFs (principal components) \cite{hron16}, which are mutually orthogonal elements $\lbrace \phi_j \rbrace_{j \geq 1}$, $\phi_j \in \B^2$,
that maximize the following objective functional
\begin{eqnarray} \label{eq:max1b}
    \frac{1}{N} \sum_{n=1}^{N} \langle f_n, \phi_j \rangle_{\B^2}^2 \quad \text{subject to} \ \| \phi_j \|_{\B^2} = 1  \ \text{and} \ \langle \phi_j, \phi_k \rangle_{\B^2} = 0, \ k < j, \ j \geq 2.
\end{eqnarray}

The solution of this maximization can be found as the eigenfunctions of the sample covariance operator $V \! : \B^2 \rightarrow \B^2$ that, to each $f \in \B^2$, assigns 
\begin{equation*}
    V f = \frac{1}{N} \odot \bigoplus_{n=1}^{N} \langle f_n, f \rangle_{\B^2} \odot f_n.
\end{equation*}
The corresponding eigenvalues are $\lbrace \rho_j \rbrace_{j \geq 1}, \ \rho_1 \geq \rho_2 \geq \cdots \geq 0$ (in the sample case, only $N$ first eigenvalues are considered). Note that the sample covariance operator differs essentially from the sample Bayes covariance matrix (\ref{eq:covmat}), which is defined at the object level for a vector of $P$ densities.

For computation of the eigenfunctions (functional principal components), we use the clr transformation to map this problem into the $\Lz$ space \cite{hron16}. Then we maximize the functional
\begin{eqnarray} 
\label{eq:max1}
    \frac{1}{N} \sum_{n=1}^{N} \langle \clr(f_n),\psi_j \rangle_{\Lz}^2 \quad \text{subject to} \ \| \psi_j \|_{\Lz} = 1  \ \text{and} \ \langle \psi_j, \psi_k \rangle_{\Lz} = 0, \ k < j, \ j \geq 2,
\end{eqnarray} 
where $\psi_j = \clr(\phi_j) \in \Lz$ and also $\phi_j = \clr^{-1}(\psi_j) \in \B^2$.

The solution is a collection $\lbrace \psi_j \rbrace_{j \geq 1}$, $\psi_j \in \Lz$, of the eigenfunctions of the sample covariance operator $ V_\clr \! : \Lz \rightarrow \Lz$ that, to each $f \in \Lz$, assigns
\begin{equation*}
    V_\clr f = \frac{1}{N} \sum_{n=1}^{N} \langle \clr(f_n), f \rangle_{\Lz} \cdot \clr(f_n).
\end{equation*}
Accordingly, for the FPCA scores, it holds
\begin{equation*}
    \xi_{nj} = \langle \clr(f_n), \psi_j \rangle_{\Lz} = \langle f_n, \phi_j \rangle_{\B^2}, \quad j \geq 1, \ n \in \{1,\dots,N\},
\end{equation*}
and each PDF from the sample can be expressed as
\begin{equation*}
    f_n = \bigoplus_{j=1}^{N} \xi_{nj} \odot \phi_j \quad \text{or} \quad \clr(f_n) = \sum_{j=1}^{N} \xi_{nj} \psi_j.
\end{equation*}


\subsection{FPCA for orthogonal parts of multivariate densities} \label{sec:mfpca}

In this section, we address the question of how multivariate FPCA for orthogonally decomposed multivariate PDFs relates to FPCA applied directly to these PDFs, demonstrating the equivalence between the two approaches. This implies that the functional principal components (i.e., the eigenfunctions) closely follow the orthogonal decomposition of the PDFs. Consequently, we can gain more insight into the data structure by viewing PDFs through the lens of their reduced dimensionality, and obtain the decomposition of the FPCA scores according to the PDF decomposition.

Recall that for $d\geq2$, all PDFs can be decomposed into $P = 2^d-1$ orthogonal parts, cf. equations (\ref{eq:decomp1}),(\ref{eq:decomp2}) and (\ref{eq:dec1.clr}),(\ref{eq:dec2.clr}), that can be expressed as
\begin{equation*}
    f_n = \bigoplus_{i=1}^d f_n^i \oplus \bigoplus_{I\subseteq D,|I|\geq2} f_n^{I,\intr} = \bigoplus_{i\in\mathcal{I}} f_n^i
\end{equation*}
and
\begin{equation*}
    \clr(f_n) = \sum_{i=1}^d \clr(f_n^i) + \sum_{I\subseteq D,|I|\geq2} \clr(f_n^{I,\intr}) = \sum_{i\in\mathcal{I}}\clr(f_n^i).
\end{equation*}

For multivariate FPCA, we consider vectors of densities from the orthogonal decomposition $\mathbf{f}_n = (f_n^i)_{i\in\mathcal{I}} \in \bigtimes_{i\in\mathcal{I}} \B^2_i \subset  [\B^2]^P$, where $n \in \{1,\dots,N\}$. The densities within each of the vectors are mutually orthogonal and different parts of the decomposition are also orthogonal across samples, i.e., $f_n^i$ and $f_m^j$, $n,m \in \{1,\dots,N\}$, $i,j \in \mathcal{I}$, $j \neq i$, are orthogonal. And, in a similar manner as in Section~\ref{sec:fpca}, we would like to find the main modes of variability of these vectors \cite{filzmoser21}. That is, we want to obtain eigenfunctions in the same vector form, with the hope that elements of this vector will follow the orthogonality of the PDF decomposition.

The space $[\B^2]^P = \B^2 \times \cdots \times \B^2$ is a Hilbert space with component-wise operations perturbation and powering, and the inner product as defined in \cite{filzmoser21}. 

\begin{definition}
   For arbitrary $\mathbf{f}=(f_p)_{p=1}^P,\, \mathbf{g}=(g_p)_{p=1}^P \in [\B^2]^P$ and $\alpha \in \mathbb{R}$,
    \begin{equation*}
       (\mathbf{f} \oplus \mathbf{g})_p = f_p \oplus g_p \quad \text{and} \quad (\alpha \odot \mathbf{f})_p = \alpha \odot f_p, \quad p \in \{1,\dots,P\}.
    \end{equation*}
    Further, the inner product is defined as
    \begin{equation*}
        \left\langle \mathbf{f},\mathbf{g} \right\rangle_{[\B^2]^P} = \sum_{p=1}^P \langle f_p,g_p \rangle_{\B^2}.
    \end{equation*}
\end{definition}

\noindent Moreover, the clr transformation and its inverse are also defined component-wise so that
\begin{eqnarray*}
    \boldsymbol{\clr}(\mathbf{f}) &=& (\clr(f_p))_{p=1}^P \in [\Lz]^P, \quad \text{for} \ \mathbf{f} = (f_p)_{p=1}^P \in [\B^2]^P, \\
    \boldsymbol{\clr}^{-1}(\mathbf{g}) &=& (\clr^{-1}(g_p))_{p=1}^P \in [\B^2]^P, \quad \text{for} \ \mathbf{g} = (g_p)_{p=1}^P \in [\Lz]^P.
\end{eqnarray*}
	
To obtain the $j$th principal component (PC), we maximize the following functional over $[\B^2]^P$, analogously to multivariate FPCA for functional data \cite{happ18},
\begin{eqnarray} \label{eq:functional}
    \frac{1}{N} \sum_{n=1}^{N} \langle \mathbf{f}_n, \bs{\phi}_j \rangle_{[\B^2]^P}^2 \quad \text{subject to} \ \| \bs{\phi}_j \|_{[\B^2]^P} = 1  \ \text{and} \ \langle \bs{\phi}_j, \bs{\phi}_k \rangle_{[\B^2]^P} = 0, \ k < j, \ j \geq 2.
\end{eqnarray}

The PCs $\bs{\phi}_j \in [\B^2]^P$ are eigenfunctions of the sample covariance operator $\bs{V} \! : [\B^2]^P \rightarrow [\B^2]^P$ that, to each $\mathbf{f} \in [\B^2]^P$, assigns
\begin{equation} \label{eq:Vf}
    \bs{V}\mathbf{f} = \frac{1}{N} \odot \bigoplus_{n=1}^{N} \langle \mathbf{f}_n, \mathbf{f} \rangle_{[\B^2]^P} \odot \mathbf{f}_n.
\end{equation}
Note that these eigenfunctions must lie in the span of $\bs{V}$ and due to  \eqref{eq:Vf} can be written as a linear combination (in the Bayes space sense) of the $\mathbf{f}_n$. Thus, like the $\mathbf{f}_n$, they  are in fact elements of the subspace $
\bigtimes_{i\in\mathcal{I}}\B^2_i \subset [\B^2]^P$.

After clr transformation of the PCA problem (\ref{eq:functional}), we maximize
\begin{eqnarray} \label{eq:max2}
    \frac{1}{N} \sum_{n=1}^{N} \langle \bs{\clr}(\mathbf{f}_n), \bs{\psi}_j \rangle_{[\Lz]^P}^2 \quad \text{subject to} \ \| \bs{\psi}_j \|_{[\Lz]^P} = 1  \ \text{and} \ \langle \bs{\psi}_j, \bs{\psi}_k \rangle_{[\Lz]^P} = 0, \ k < j, \ j \geq 2,
\end{eqnarray}
where $\bs{\psi}_j = \bs{\clr}(\bs{\phi}_j) \in \bigtimes_{i\in\mathcal{I}} L^2_{0,i} \subset [\Lz]^P$ and $\bs{\phi}_j \in [ \B^2 ]^P$.
	
Then, to each $\mathbf{f} \in [\Lz]^P$, the transformed sample covariance operator $\bs{V_\clr} \! : [\Lz]^P \rightarrow [\Lz]^P$ assigns
\begin{equation*}
	\bs{V_\clr} \mathbf{f} = \frac{1}{N} \sum_{n=1}^{N} \langle \bs{\clr}(\mathbf{f}_n), \mathbf{f} \rangle_{[\Lz]^P} \cdot \bs{\clr}(\mathbf{f}_n).
\end{equation*}
       
The PCs $\bs{\phi}_j, \, j \geq 1$ are orthonormal, and for each $j \geq 1$, $\phi_j^i,\, i\in\mathcal{I},$ are mutually orthogonal. However, $\lbrace \phi^i_j \rbrace_{j \geq 1}$ alone are not orthonormal nor orthogonal for any component $i\in\mathcal{I}$. Analogously, $\lbrace \bs{\psi}_j \rbrace_{j \geq 1}$ has the same properties.
	
For vectors of PDFs, it holds that
\begin{equation*}
    \mathbf{f}_n = \bigoplus_{j=1}^{N} \zeta_{nj} \odot \bs{\phi}_j \quad \text{and} \quad \bs{\clr}(\mathbf{f}_n) = \sum_{j=1}^{N} \zeta_{nj} \bs{\psi}_j
\end{equation*}
for $n \in \{1,\dots,N\}$, where $\zeta_{nj} = \langle \bs{\clr}(\mathbf{f}_n), \bs{\psi}_j \rangle_{[\Lz]^P} = \langle \mathbf{f}_n, \bs{\phi}_j \rangle_{[\B^2]^P}$ are the respective scores. 


We now present an interesting equivalence between FPCA for multivariate densities (Section~\ref{sec:fpca}) and multivariate FPCA for vectors of decomposed densities (Section~\ref{sec:mfpca}) that can be employed to conduct a more in-depth analysis of FPCA results (eigenfunctions and scores) for multivariate densities. While the following theorem is formulated and proved in the clr space, the corresponding result holds in the Bayes space as well, as discussed below.

\begin{theorem} \label{theo:equiv}
    Let $\lbrace \psi_j^0 \rbrace_{j\geq1}$ be the principal components maximizing the FPCA functional for multivariate densities on clr level in (\ref{eq:max1}), and let $\bs\psi_j = (\psi_j^i)_{i\in\mathcal{I}},\, j\geq1$ be the vectors of principal components maximizing the multivariate FPCA functional for decomposed densities on clr level (\ref{eq:max2}). Then the functional principal components decompose as
    \begin{equation*}
        \psi_j^0 = \sum_{i\in\mathcal{I}} \psi_j^i \quad \text{for all} \ j\geq 1
    \end{equation*}
    and consequently
    \begin{equation}
        \label{eq:clr.scores}
        \xi_{nj} = \langle \clr(f_n),\psi_j^0 \rangle_{\Lz} = \sum_{i\in \mathcal{I}} \langle \clr(f^i_n),\psi^i_j \rangle_{\Lz} = \langle \bs{\clr}(\mathbf{f}_n), \bs{\psi}_j \rangle_{[\Lz]^P} = \zeta_{nj},
\end{equation}
where $j \geq 1, \ n \in \{1,\dots,N\}$.
\end{theorem}

\begin{proof}[\textbf{\upshape Proof:}]
    For all $j \geq 1$, let $\psi_j^* \in \Lz$ maximize (\ref{eq:max1}) and $\tilde{\bs{\psi}_j} = (\tilde{\psi}_j^i)_{i\in\mathcal{I}} \in \bigtimes_{i\in\mathcal{I}} L^2_{0,i} \subset [\Lz]^P$ maximize (\ref{eq:max2}). We show the equivalence of the maximizers and the corresponding constraints in (\ref{eq:max1}) and (\ref{eq:max2}), and therefore the equality $\psi_j^* = \sum_{i\in\mathcal{I}} \tilde{\psi}_j^i,\ \forall j\geq 1$. 

    First, we show the equality of the inner products in (\ref{eq:max1}) and (\ref{eq:max2}) in the following sense. For any $u,v \in \Lz,\ u = \sum_{i\in\mathcal{I}} u^i,\ v = \sum_{i\in\mathcal{I}} v^i$ using the unique decomposition in (\ref{eq:dec1.clr}), and $\mathbf{u},\mathbf{v} \in \bigtimes_{i\in\mathcal{I}}L^2_{0,i}$ with $\mathbf{u} = (u^i)_{i\in\mathcal{I}},\ \mathbf{v} = (v^i)_{i\in\mathcal{I}}$, we can write
    \begin{equation}
        \label{ip:2} \langle u,v \rangle_{\Lz} = \left\langle \sum_{i\in\mathcal{I}} u^i, \sum_{i\in\mathcal{I}} v^i \right\rangle_{\Lz} = \sum_{i\in\mathcal{I}} \langle u^i,v^i \rangle_{\Lz} = \langle \mathbf{u},\mathbf{v} \rangle_{[\Lz]^P},
    \end{equation}
    where the second equality in (\ref{ip:2}) holds because of the orthogonality of the spaces $L^2_{0,i},i\in\mathcal{I}$, which means that $u^i$ and $v^j$, $i,j \in \mathcal{I}$, $j \neq i$, are orthogonal.

    Thus, with $\psi_j = \sum_{i\in\mathcal{I}} \psi_j^i$, $\bs{\psi}_j = (\psi_j^i)_{i\in\mathcal{I}}$, $j\geq 1$, it holds that
    \begin{equation}\label{equiv1}
        \| \psi_j \|_{\Lz} = 1 \quad \Leftrightarrow \quad \| \bs{\psi}_j \|_{[\Lz]^P} = 1,
    \end{equation}
    \begin{equation}\label{equiv2}
        \langle \psi_j, \psi_k \rangle_{\Lz} = 0 \quad \Leftrightarrow \quad \langle \bs{\psi}_j, \bs{\psi}_k \rangle_{[\Lz]^P} = 0
    \end{equation}
    and
    \begin{equation}\label{equiv3}
        \psi_j^* = \text{argmax}\, \frac{1}{N} \sum_{n=1}^{N} \langle \clr(f_n),\psi_j \rangle_{\Lz}^2 \ \Leftrightarrow \ \bs{\psi}_j^* = \text{argmax}\, \frac{1}{N} \sum_{n=1}^{N} \langle \bs{\clr}(\mathbf{f}_n), \bs{\psi}_j \rangle_{[\Lz]^P}^2,
    \end{equation}
    with $\psi_j^* = \sum_{i\in\mathcal{I}} (\psi_j^*)^i $ and $ \bs{\psi_j}^* = \left( (\psi_j^*)^i \right)_{i\in\mathcal{I}}$
    as
    \begin{equation}\label{equiv4}
        \frac{1}{N} \sum_{n=1}^{N} \langle \clr(f_n),\psi_j \rangle_{\Lz}^2 = \frac{1}{N} \sum_{n=1}^{N} \langle \bs{\clr}(\mathbf{f}_n), \bs{\psi}_j \rangle_{[\Lz]^P}^2,
    \end{equation}
    where (\ref{equiv1}), (\ref{equiv2}), (\ref{equiv4}) follow from the equality of the inner products (\ref{ip:2}).

    As $\tilde{\bs{\psi}}_j$ was the maximizer of the right-hand side of (\ref{equiv3}) under the same constraints, altogether, this shows that $\psi_j^* = \sum_{i\in\mathcal{I}} \tilde{\psi}_j^i,\ j\geq 1$.

    The equality of scores $\xi_{n,j}$ and $\zeta_{n,j}$ (\ref{eq:clr.scores}) also follows from (\ref{ip:2}).  
\end{proof}

As stated before, the equivalence is shown only in the $\Lz$ space, but it also holds in the Bayes space $\B^2$, which can be shown as a direct consequence of Theorem \ref{theo:equiv} and the isometry of the clr transformation.

\begin{corollary} \label{coro:equiv}
    Let $\lbrace \phi_j^0 \rbrace_{j\geq1}$ be the principal components maximizing the functional in (\ref{eq:max1b}), and let $\bs\phi_j = (\phi_j^i)_{i\in\mathcal{I}},\, j\geq1$, be the vectors of principal components maximizing (\ref{eq:functional}). Then 
    \begin{equation}
    \label{eq:b.loadings}
        \phi_j^0 = \bigoplus_{i\in\mathcal{I}} \phi_j^i \quad \text{for all} \ j\geq 1,
    \end{equation}
    and
    \begin{equation}
    \label{eq:b.scores}
        \xi_{nj} = \langle f_n, \phi_j^0 \rangle_{\B^2} = \sum_{i\in\mathcal{I}} \langle f^i_n,\phi^i_j \rangle_{\B^2} = \langle \mathbf{f}_n, \bs{\phi}_j \rangle_{[\B^2]^P} = \zeta_{nj},
    \end{equation}
    where $j \geq 1, \ n \in \{1,\dots,N\}$, and further
    \begin{equation} \label{eq:final.result}
    f_n = \bigoplus_{i\in\mathcal{I}} f_n^i = \bigoplus_{j=1}^{N} \xi_{ij} \odot \left( \bigoplus_{i\in\mathcal{I}} \phi_j^i \right).
\end{equation}
\end{corollary}

Equations (\ref{eq:clr.scores}) and (\ref{eq:b.scores}) also show the decomposition of the scores into the sum of $P$ elements, each corresponding to one of the parts from the orthogonal decomposition of the PDFs.

Equation (\ref{eq:final.result}) suggests how to proceed with the practical computations. The scores and eigenfunctions of FPCA for multivariate densities can be decomposed using Equations (\ref{eq:b.loadings}) and (\ref{eq:b.scores}). This means that the eigenfunctions can be decomposed directly using (\ref{eq:decomp1}) and (\ref{eq:decomp2}), and the scores are computed using inner products with the decomposed PDFs. Then, Theorem \ref{theo:equiv} and Corollary \ref{coro:equiv} guarantee that this orthogonal decomposition of eigenfunctions is unique. Basically, we show that this computational process is equivalent to first decomposing the densities and then performing mFPCA.

Decomposing eigenfunctions and scores paves the way for diagnostics in FPCA for multivariate densities. It is possible to evaluate the relevance and particular structure of the scores and eigenfunctions for (geometric) marginal PDFs, as well as for those related to the interaction between variables.

\begin{remark}
    The present paper focuses on the structural properties of the orthogonal decomposition and its implications for dimension reduction. Nevertheless, the asymptotic behaviour of the proposed procedure follows from standard results for functional principal component analysis in Hilbert spaces. Indeed, the Bayes space $\mathcal{B}^2(\lambda)$ is isometrically isomorphic to the Hilbert space $L_0^2(\lambda)$ through the clr transformation, such that FPCA for multivariate densities can be studied through the corresponding covariance operator in $L_0^2(\lambda)$. Under standard assumptions for functional PCA, e.g., existence of a square-integrable random process on clr-transformed density level and suitable eigenvalue separation, the empirical covariance operator converges to its population counterpart as $N \to \infty$, implying consistency of the estimated eigenvalues, eigenfunctions and scores; see \cite{kokoszka17}, Chapter 12. Since the orthogonal decomposition into geometric marginals and interaction components is defined via orthogonal projections onto closed subspaces, the decomposition inherits these asymptotic properties. Consequently, the variance decomposition and the decomposition of eigenfunctions and scores derived in Sections \ref{sec:pca} and \ref{sec:fpcadec} also hold at the population level and can be consistently estimated from the sample.
\end{remark}


\section{Applications}
\label{sec:app}

Two empirical data sets were analyzed to illustrate the usefulness of the orthogonal decomposition in dimension reduction of PDFs using FPCA. We focus on bivariate densities, which can still be easily visualized, and demonstrate how the decomposition of scores and loadings, as well as the decomposition of the variance, can enhance the analysis. For example, one can better see whether the variability comes from single variables (in terms of their geometric marginals) or from their interaction, and how each component of the decomposition contributes to the overall (reduced) data structure. The first application focuses on distributions of house price and house size for all 50 US states, the District of Columbia, and Puerto Rico. In the second application, geological data are analyzed, specifically PDFs of copper and zinc concentrations in 77 Czech districts. The results of these applications, along with the computational details, are discussed below. There is also included a sensitivity analysis to examine the effect of replacement of zero density values. All computations and plots were performed using the software \textsf{R} \cite{r}, the source code is available at \url{https://github.com/adelaczolkova/DimReductionOfDensities}. 


\subsection{Housing data}
\label{housing}

This analysis focuses on US housing data, real estate data downloaded from Kaggle \citep{housing}. We used data for all 50 US states, the District of Columbia, and Puerto Rico (hereafter referred to as 'states'), which had sufficient observations to derive the respective PDFs, and we filtered out observations with missing values.

Of the variables available in the data set, the house price in dollars ($x$~variable) and the house size in square feet ($y$ variable) were considered for the analysis. The domains of these two variables were truncated to sensible intervals of prices and sizes, covering domains in all states, specifically $(2000,2\cdot10^8)$ dollars for price and $(150,10^5)$ square feet for size. Therefore, the final version of the dataset contains 1 655 633 observations -- the number of observations for each state ranges from 969 (Alaska) to 197 544 (California). Due to truncation, 853 observations (0.05~\% of the data) were lost. Both variables were log-transformed. Following \cite{MatysGrygar2024}, a kernel density estimator was used to obtain bivariate densities of price and size for each state, resulting in discretized densities on a 30$\times$30 grid. The next step was the clr transformation accompanied by an imputation of zeros, which would cause computational problems due to the use of logarithms (zeros were imputed by a constant value similarly as in \cite{MatysGrygar2024}). The clr-transformed bivariate densities ($\clr(f_n)$) were also decomposed into three orthogonal parts -- two geometric marginals ($\clr(f_n^1) \equiv \clr(f_n^\X)$ and $\clr(f_n^2) \equiv \clr(f_n^\Y)$), which together form the independent part ($\clr(f_n^\ind)$), and one interactive part ($\clr(f_n^\intr)$), which represents the relationship between the two variables.

\begin{figure}[ht]
    \centering
    \includegraphics[width=0.8\linewidth]{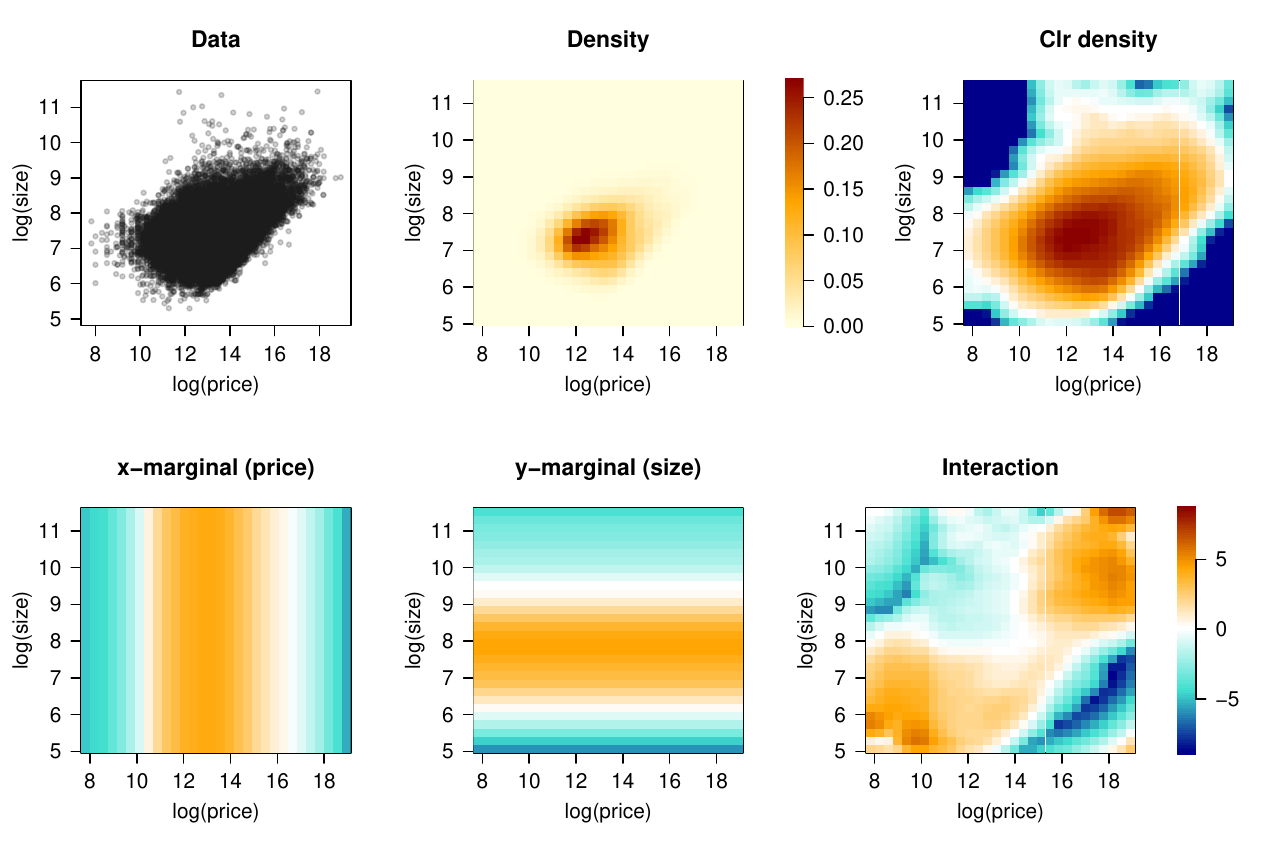}
    \caption{New York state: The first row shows the raw data, the density estimate and the clr-transformed density. The second row shows the orthogonal decomposition of this density on clr level (the two geometric marginals and the interactive part).}
    \label{fig:NY}
\end{figure}

Fig.~\ref{fig:NY} demonstrates the results of kernel density estimation and decomposition for one state -- New York. The upper row shows the raw data for New York state in the first panel, which were used to estimate the density (in $\B^2$) displayed in the second panel, and the third panel shows the clr-transformed density (in $\Lz$). The three plots in the bottom row show the orthogonal decomposition -- the geometric marginals for price and size, respectively, and the interactive part (all in $\Lz$).

The clr-transformed densities ($\clr(f_n)$) were centered and then the FPCA for multivariate densities (Section~\ref{sec:fpca}), specifically bivariate PDFs here, was performed resulting in scores $\xi_{nj}$ and loadings (eigenfunctions) $\psi_j \in \Lz$, where the indices $n$ correspond to states and $j$ to components. The eigenfunctions were directly decomposed into $\bs{\psi}_j =(\psi_j^\X,\psi_j^\Y,\psi_j^\intr)^\top \in [\Lz]^3$, which could be obtained alternatively from mFPCA for vectors of PDFs from the orthogonal decomposition, i.e., $\bs{\clr}(\mathbf{f}_n) = (\clr(f_n^\X),\clr(f_n^\Y),\clr(f_n^\intr))^\top$ (Section~\ref{sec:mfpca}). The decomposed eigenfunctions were then used to compute the decomposed scores $\langle \clr(f_n^\X),\psi_j^\X \rangle_{\Lz}, \langle \clr(f_n^\Y),\psi_j^\Y \rangle_{\Lz}, \langle \clr(f_n^\intr),\psi_j^\intr \rangle_{\Lz}$. 

\begin{figure}[t]
    \centering
    \includegraphics[width=0.95\linewidth]{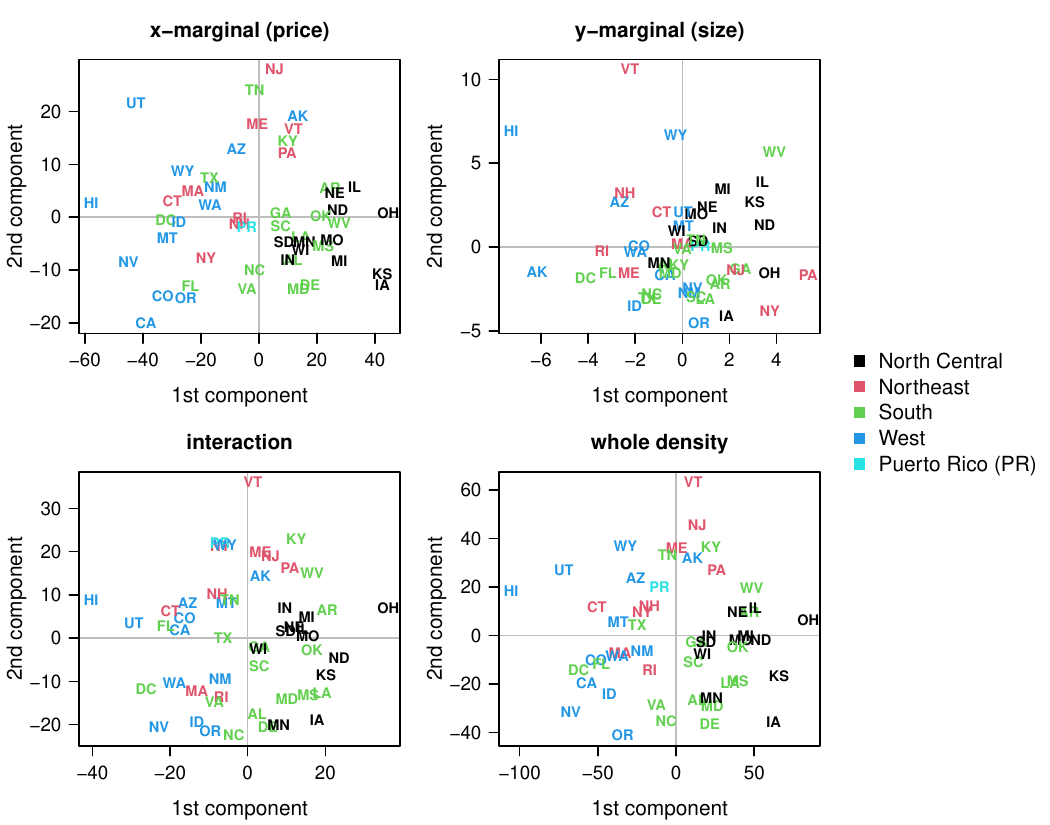}
    \caption{Housing scores and their decomposition with colours according to the four main US regions and Puerto Rico, which does not belong to any of these regions.}
    \label{fig:h:scores}
\end{figure}

Fig.~\ref{fig:h:scores} shows the resulting scores and their decomposition corresponding to the first two functional principal components. Some structural differences between the US regions are visible (North Central, Northeast, South, and West; Puerto Rico does not belong to any of the regions listed). These differences are mostly related to price and interaction scores (left plots). For example, one can see that the scores of the North Central states corresponding to the first component are always positive in these two plots, which also holds for the total scores (bottom-right plot). Furthermore, Vermont (VT) and Hawaii (HI) can be considered as outliers. This fact is mainly connected to the scores for the size variable (top-right plot) and the interaction (bottom-left plot). However, regarding the price scores (top-left plot), Vermont (VT) is part of the data cloud, meaning that it does not deviate in terms of price alone. However, when price and size are considered together, we can see that Vermont is different from the other states. Thus, the decomposition provides a deeper insight into the density data structure, which clearly demonstrates its usefulness in the exploratory functional data analysis of multivariate PDFs.

\begin{figure}[t]
    \centering
    \includegraphics[width=0.95\linewidth]{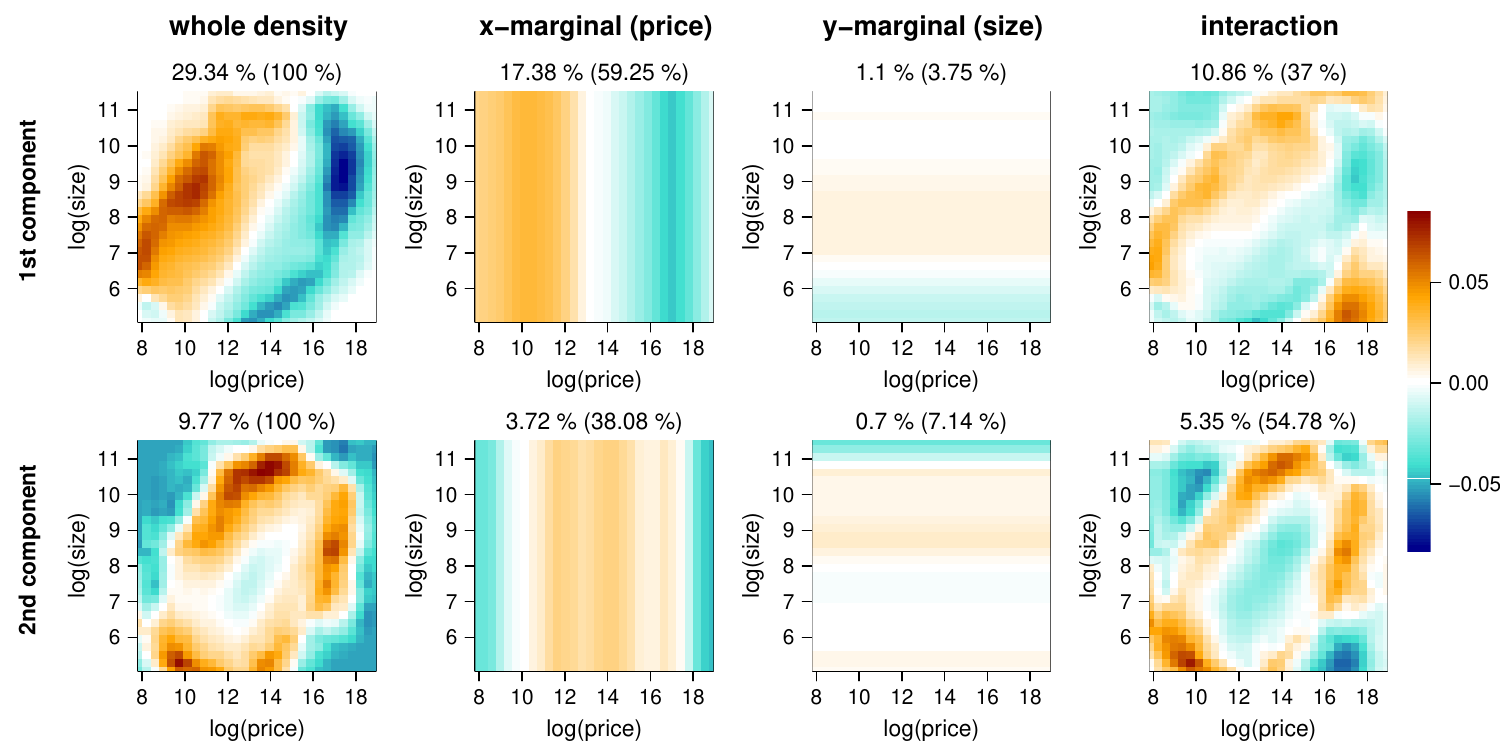}
    \caption{Eigenfunctions of US housing data for the first two functional principal components and their orthogonal decomposition on clr level. The percentages show how much variance is explained by each eigenfunction or its part. The percentages in brackets show how much variance of the functional principal component are explained by the particular parts of the eigenfunction.}
    \label{fig:h:loadings}
\end{figure}

Looking at the eigenfunctions of the first two components and their decomposition (Fig.~\ref{fig:h:loadings}) can also help with the interpretation of the scores. For example, the first component for the price ($x$-marginal) is easily interpretable -- the higher the corresponding score, the cheaper the house. In the score plot (top-left plot in Fig.~\ref{fig:h:scores}), we can see that according to the first component, the cheapest houses are mostly in the North Central and some South regions. On the other hand, the most expensive are the houses in the West (which also covers Hawaii). Further, one can notice that the size ($y$-marginal) loadings are quite flat compared to the loadings corresponding to the other parts of the decomposition. This is the reason why the size score values are in general closest to zero from all the scores and therefore have the smallest influence on the structure of the total scores. While the loadings corresponding to the geometric marginals are rather easily interpretable, the interaction loadings indicate a more complex dependence structure, i.e., stronger interactions can be expected at the tails of the price and size bivariate distributions. Specifically, the first component indicates a strong interaction between high price and low size, while the second component highlights, for example, an interaction between low size and price values. 

\begin{figure}
    \centering
    \includegraphics[width=0.6\linewidth]{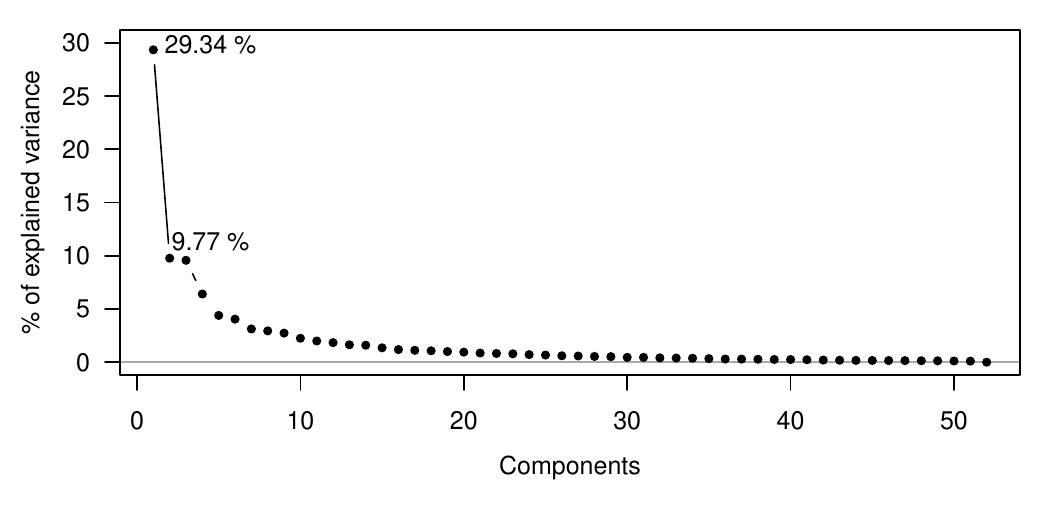}
    \caption{Scree plot for the housing data. Together, the first two principal components explain 39.11~\% of variance in the data.}
    \label{fig:h:scree}
\end{figure}

In PCA, it is useful to explore how much information (specifically variability) can be explained by the first few components. Fig.~\ref{fig:h:scree} shows how much variance is explained by the first two components that we used, and displays the results in the form of a scree plot. We can see that the percentage of explained variance for the first two components, which we used, is not very high as they explain only 39.11~\% of variance; this is, however, expected with complex socioeconomic data such as these. Moreover, there is quite a large difference between the variances of both components, as the first one explains 29.34~\% and the second one only 9.77~\% of the variance, very closely followed by the third with 9.57~\%. More details on the percentages of explained variance are included in Fig.\ref{fig:h:loadings}.

\begin{figure}[htb]
    \centering
    \includegraphics[width=0.9\linewidth]{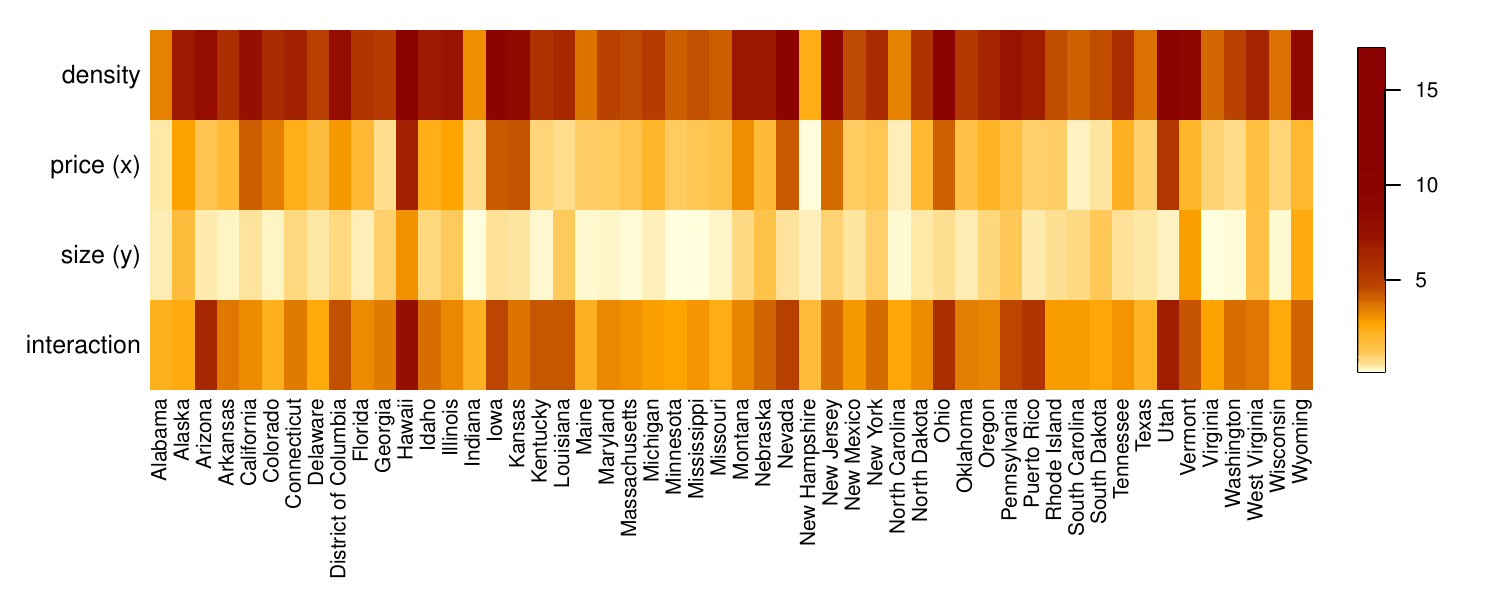}
    \caption{Squared norms of the centered bivariate densities and their decomposition for all the states.}
    \label{fig:h:norms}
\end{figure}

\begin{figure}[htb]
    \centering
    \includegraphics[width=0.9\linewidth]{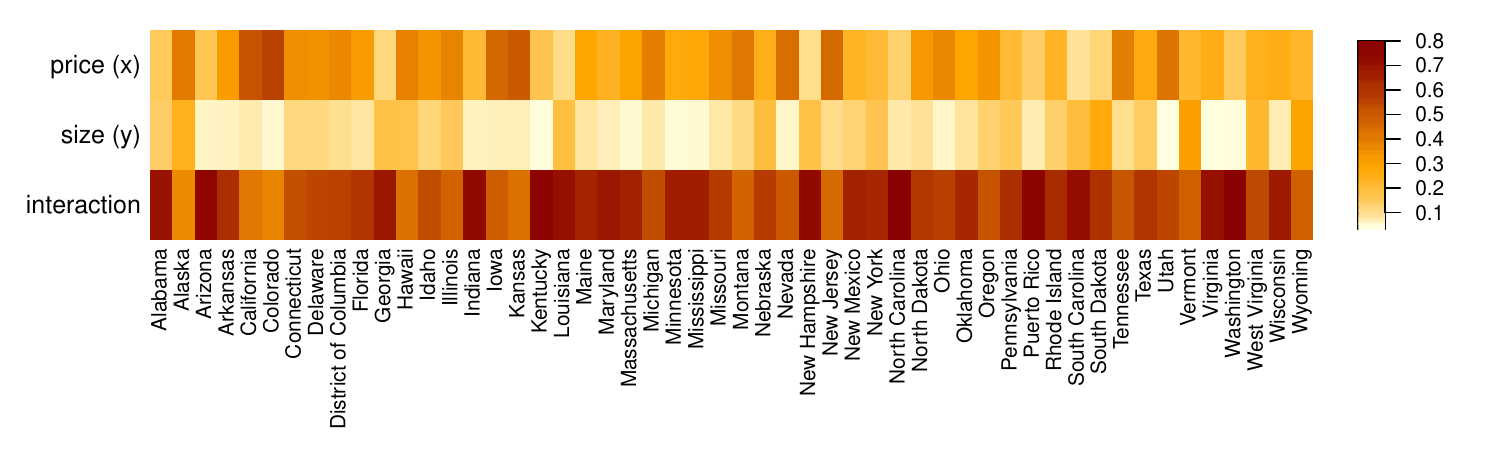}
    \caption{Relative norms of the orthogonal parts of the centered bivariate housing densities.}
    \label{fig:h:rel.norms}
\end{figure}

The orthogonal decomposition can also be used to determine how each of the density parts contributes to the total variability of the original density. Here, the variability or amount of information contained within each density (or each part) is equal to a squared norm for a centered density, and it holds that $\|f_n\|_{\B^2}^2 = \|f_n^\X\|_{\B^2}^2 + \|f_n^\Y\|_{\B^2}^2 + \|f_n^\intr\|_{\B^2}^2$, or equivalently $\|\clr(f_n)\|_{\Lz}^2 = \|\clr(f_n^\X)\|_{\Lz}^2 + \|\clr(f_n^\Y)\|_{\Lz}^2 + \|\clr(f_n^\intr)\|_{\Lz}^2$. The larger the squared norm, the larger the amount of information carried by the density. If we look at the heatmaps of squared norms (Fig. \ref{fig:h:norms}), it seems that the most relevant part of the densities is, in general, the interaction between price and house size. Some states also have a large squared norm of the price marginal. Moreover, it is worth noting that we must take into account that each density has a different total squared norm (e.g., very low in the case of Indiana or New Hampshire). This means that the decomposed variability cannot be easily compared among the states; for the comparison, it is more useful to look at the relative norms. These are shown in Fig.~\ref{fig:h:rel.norms}, where we can see that the interaction is the most informative part for most of the states.

Further, the variability of the whole sample and its decomposition can be computed as described in Section~\ref{sec:pca}. The total variance of the sample and its decomposition are the sample means of the squared norms in the case of centered densities, i.e., $\mathrm{var}_{\B^2}(\mathbf{f})=\mathrm{var}_{\B^2}(\mathbf{f}^\X) + \mathrm{var}_{\B^2}(\mathbf{f}^\Y)+ \mathrm{var}_{\B^2}(\mathbf{f}^{\intr})$, or equivalently $\mathrm{var}(\bs{\clr}(\mathbf{f})) = \mathrm{var}(\bs{\clr}(\mathbf{f}^\X)) + \mathrm{var}(\bs{\clr}(\mathbf{f}^\Y))+ \mathrm{var}(\bs{\clr}(\mathbf{f}^{\intr}))$. As the sample total variance itself does not provide any valuable information, we can focus on its decomposition and the percentages of variance contained in the orthogonal density parts. Here, the variance of the price marginals ($x$) captures 31.18~\% of the total variance, the size marginals ($y$) 11.98~\%, and the interaction 56.84~\%. Therefore, the interaction contributes the most to the total variance of the sample, and the price marginals have the second largest variance. These facts can be put in relation to the FPCA, where price and interaction components were also the most informative ones. The variance decomposition suggests that the reason is that they contain the majority of the variance in the density data.

Note that while in Section~\ref{sec:pca} we considered densities in the Bayes space $\B^2$, all the computations here were performed in the clr space $\Lz$. But this is equivalent due to the equality of inner products (and consequently squared norms) in $\B^2$ and $\Lz$.

Next, a sensitivity analysis was conducted to evaluate the robustness of the proposed dimension reduction method when handling zeros. Having performed kernel density estimation, zero density values were imputed with a constant value, in line with the approach described in \cite{MatysGrygar2024}. For each density, the constant was chosen as the mean of the ten smallest values above the cut-off limit for non-zero values, which was set at $10^{-14}$ to avoid extremely small non-zero function values from further analysis. Values smaller than the mean were imputed with the mean. This study examines how FPCA results change depending on the constant used for imputation. Taking inspiration from the simplest zero replacement strategies in compositional data analysis, we employed the mean (the default setting) and its multiples $(0.9, 0.8, 0.7, 0.6, 0.5)$ for this purpose.

\begin{figure}
    \centering
    \includegraphics[width=0.9\linewidth]{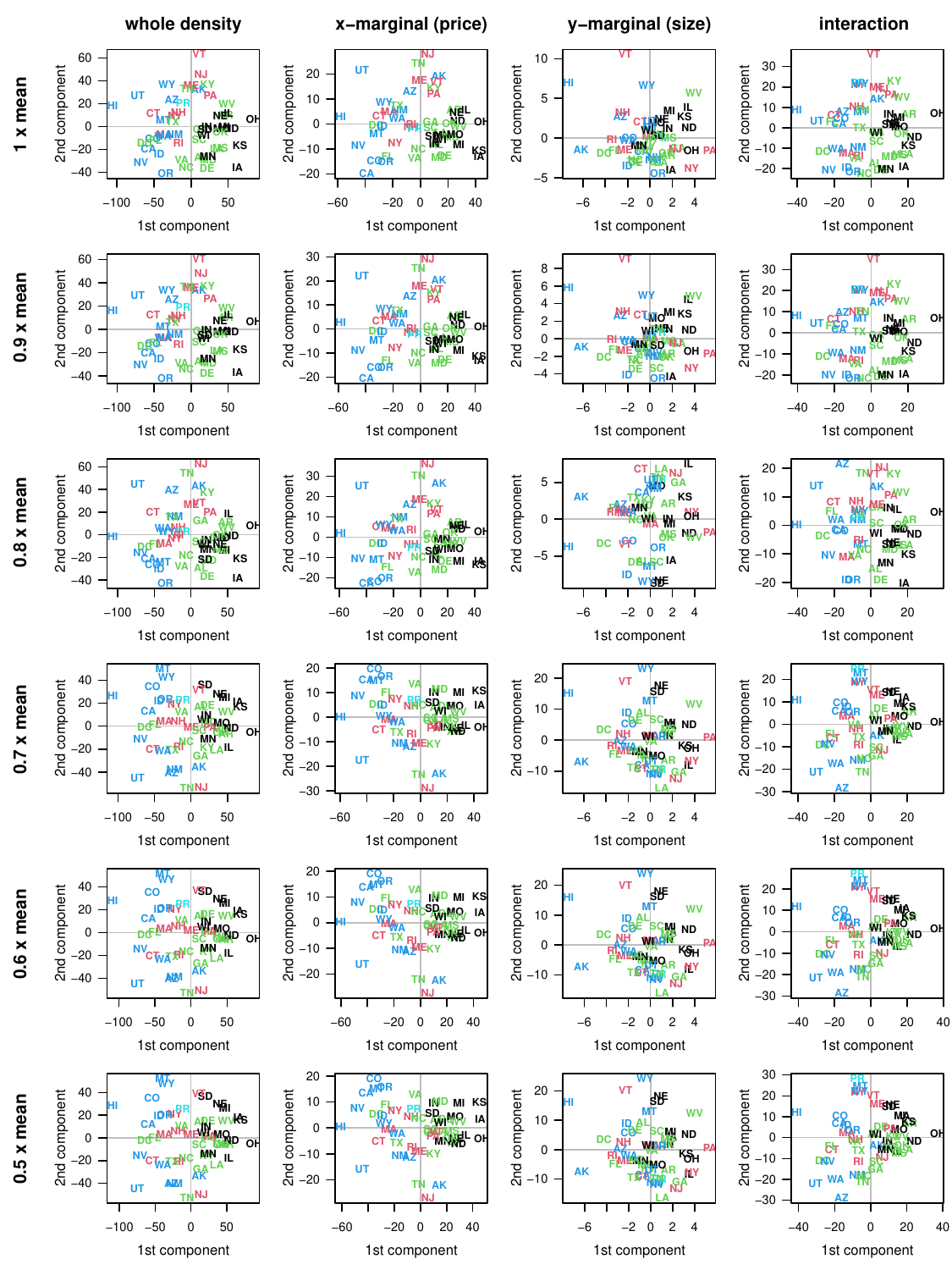}
    \caption{Decomposition of FPCA scores depending on the imputation strategy.}
    \label{fig:sim_scores}
\end{figure}

\begin{figure}
    \centering
    \includegraphics[width=0.8\linewidth]{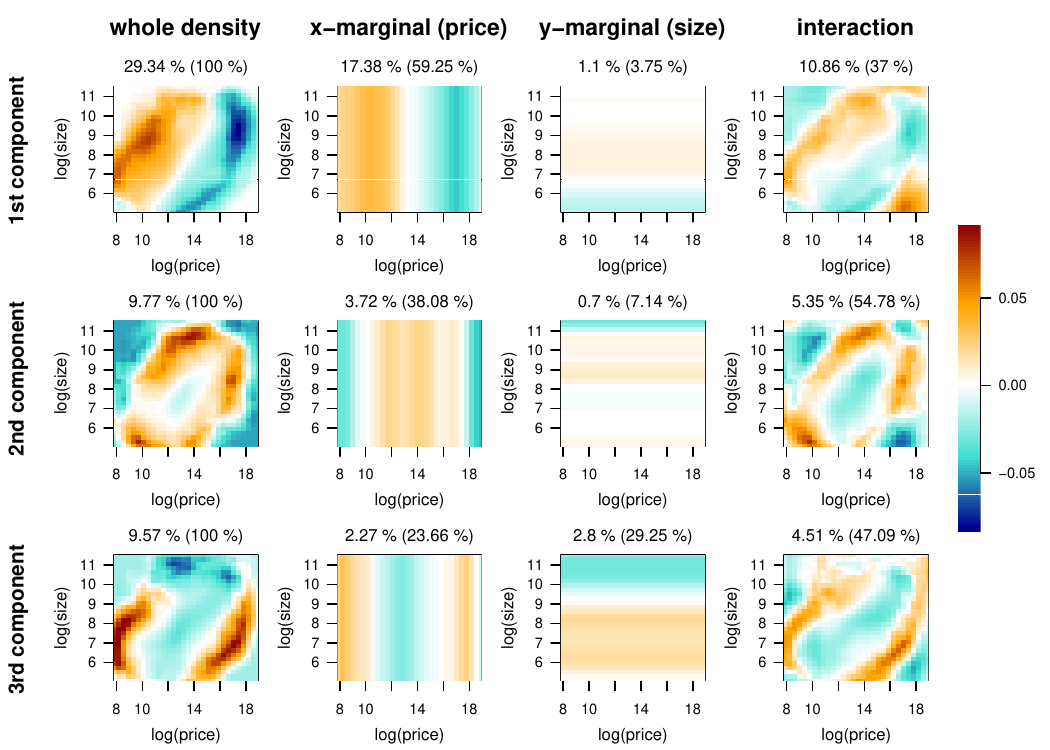}
    \caption{Loadings: $1\times\text{mean}$}
    \label{fig:sim_l1}
\end{figure}

\begin{figure}
    \centering
    \includegraphics[width=0.8\linewidth]{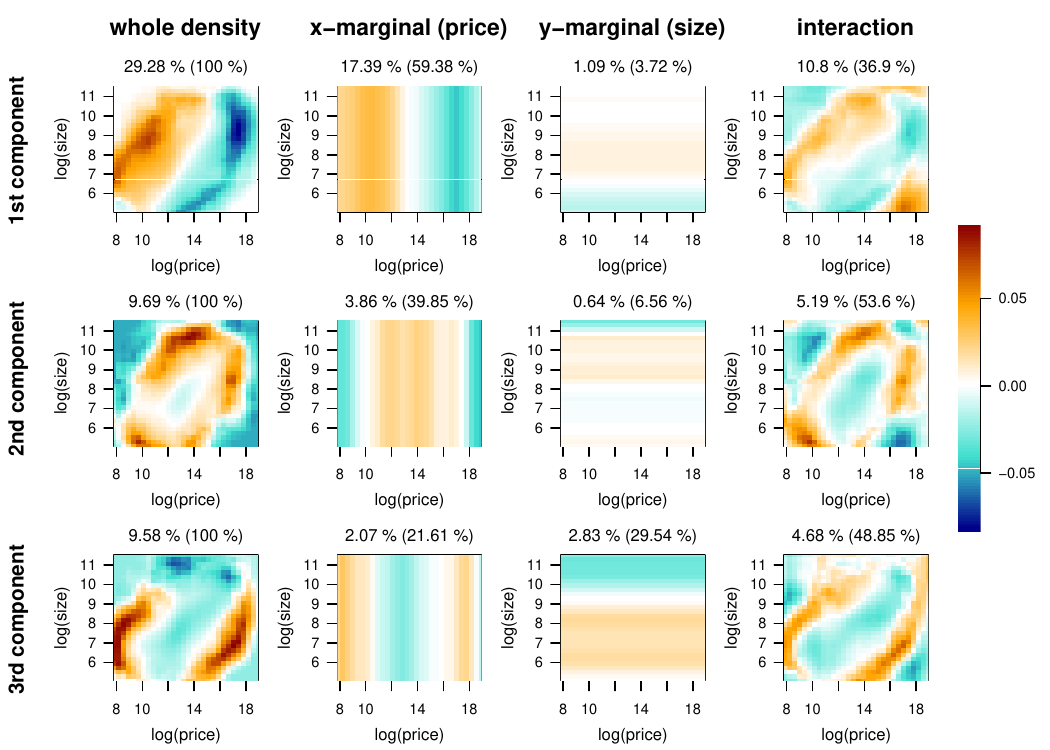}
    \caption{Loadings: $0.9\times\text{mean}$}
    \label{fig:sim_l09}
\end{figure}

\begin{figure}
    \centering
    \includegraphics[width=0.8\linewidth]{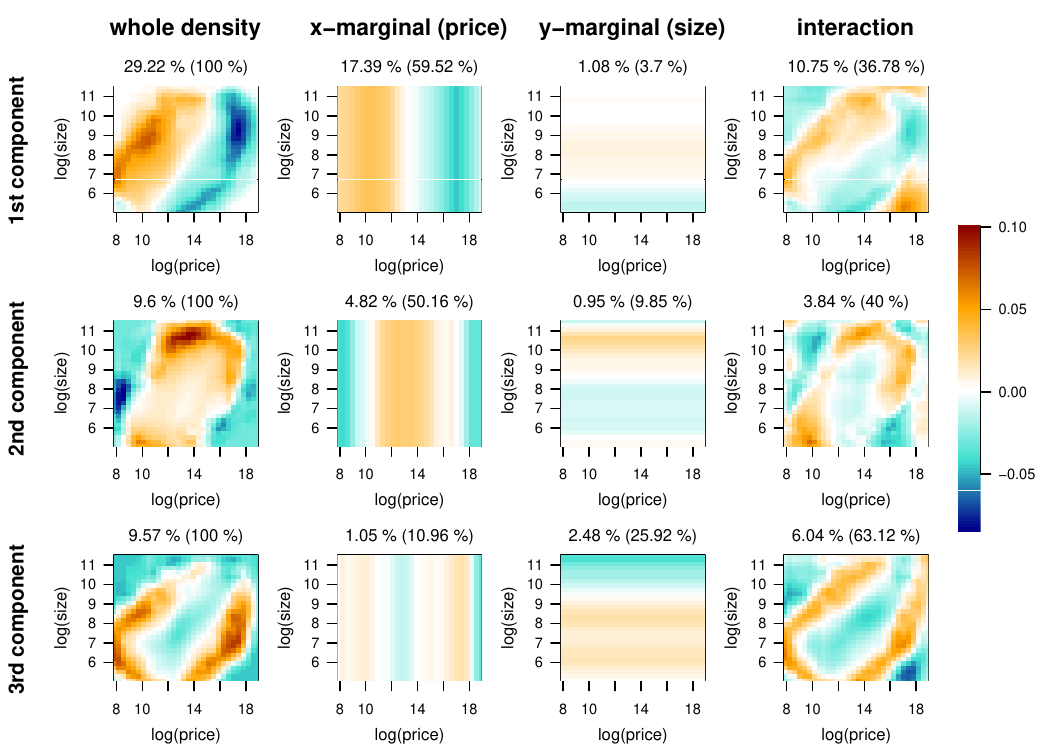}
    \caption{Loadings: $0.8\times\text{mean}$}
    \label{fig:sim_l08}
\end{figure}

\begin{figure}
    \centering
    \includegraphics[width=0.8\linewidth]{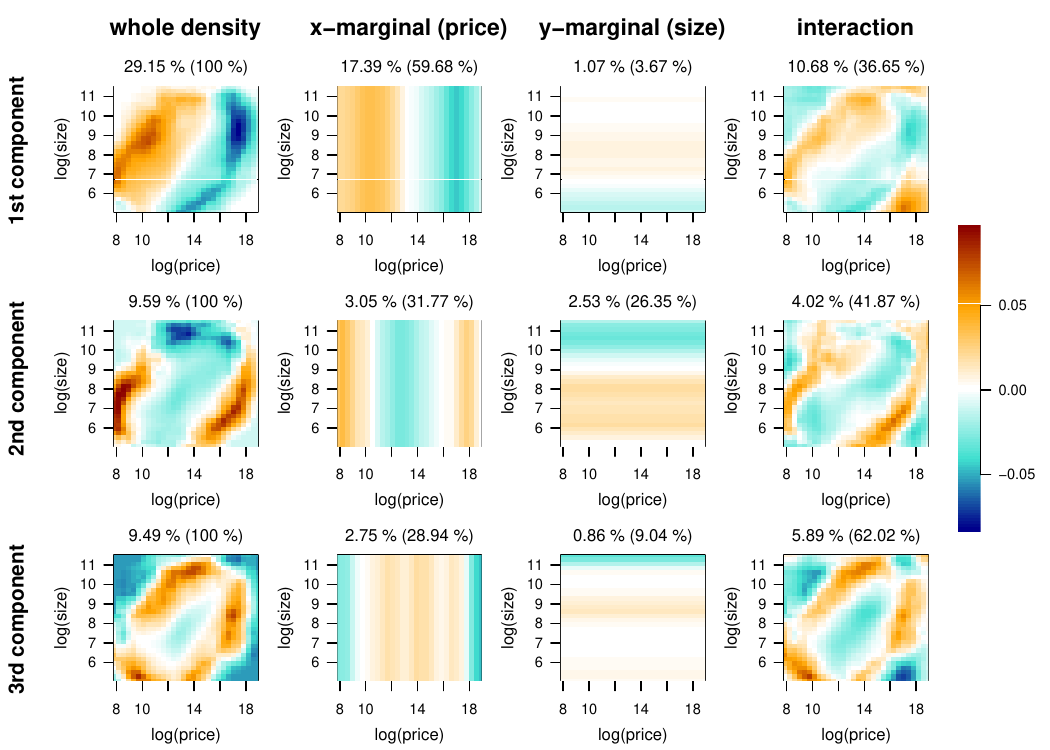}
    \caption{Loadings: $0.7\times\text{mean}$}
    \label{fig:sim_l07}
\end{figure}

\begin{figure}
    \centering
    \includegraphics[width=0.8\linewidth]{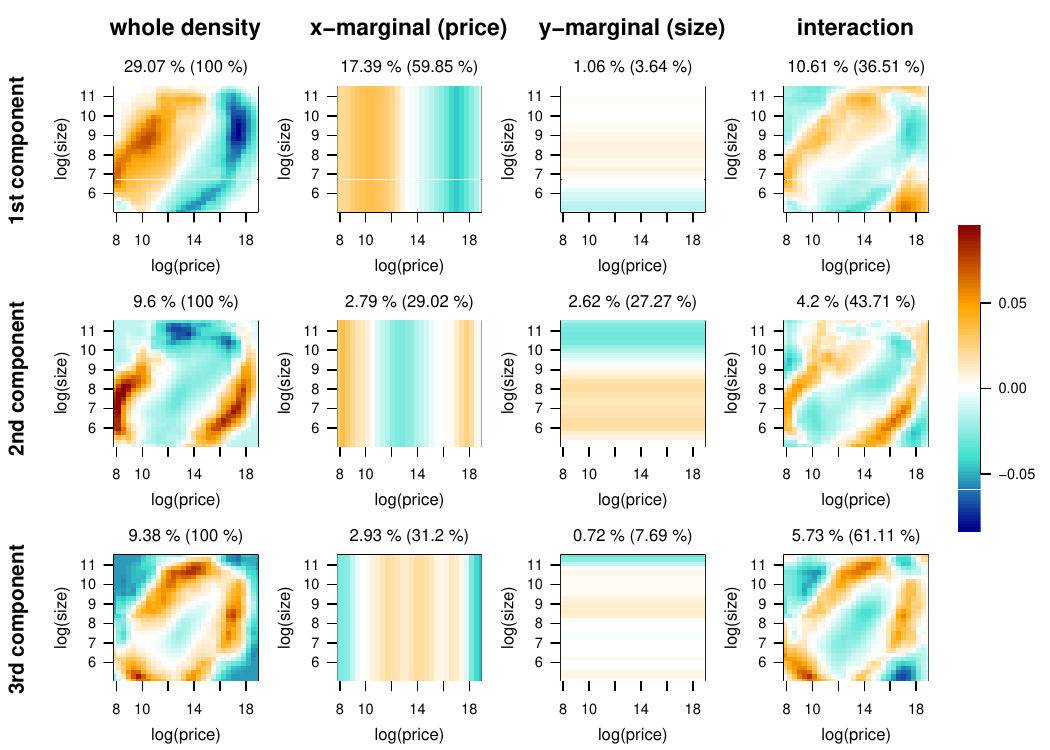}
    \caption{Loadings: $0.6\times\text{mean}$}
    \label{fig:sim_l06}
\end{figure}

\begin{figure}
    \centering
    \includegraphics[width=0.8\linewidth]{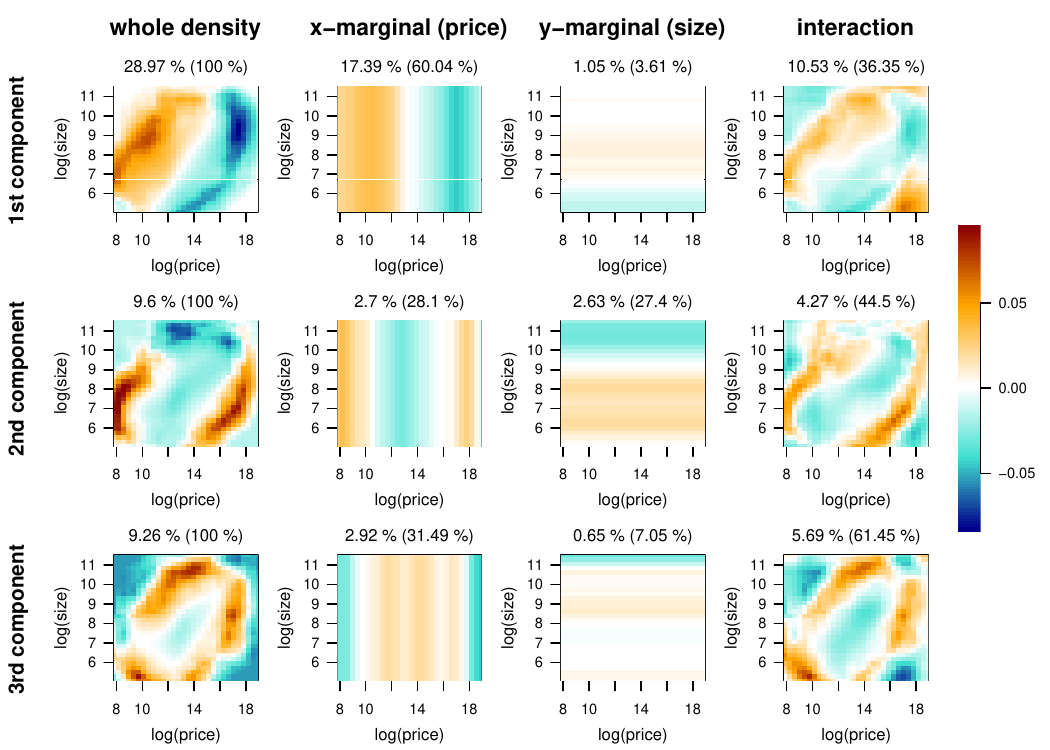}
    \caption{Loadings: $0.5\times\text{mean}$}
    \label{fig:sim_l05}
\end{figure}

\begin{figure}
    \centering
    \includegraphics[width=0.7\linewidth]{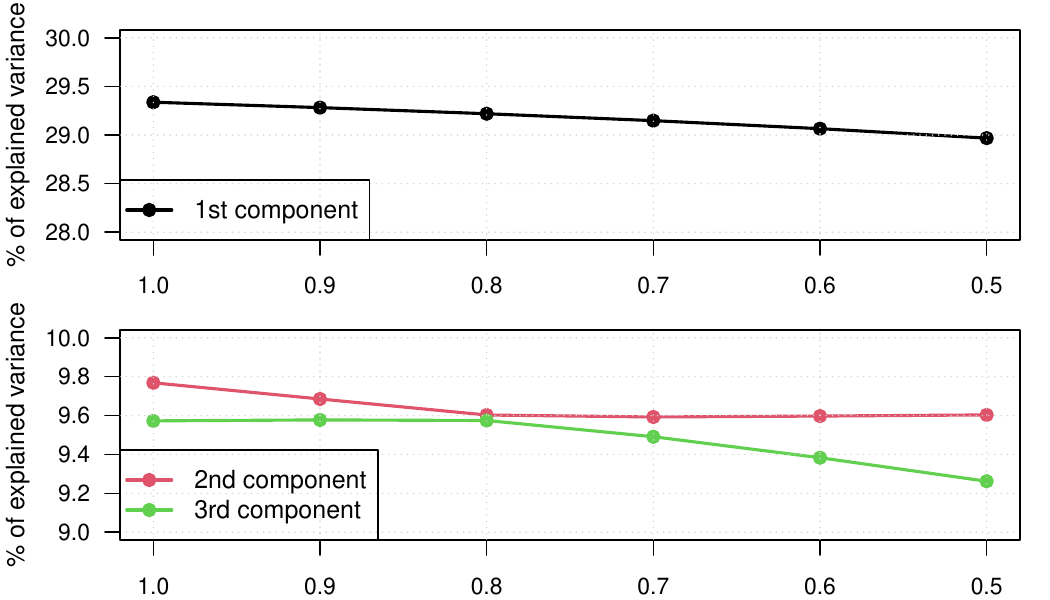}
    \caption{Percentage of explained variance by the first three components depending on the imputation constant (multiple of the default value).}
    \label{fig:sim_perc_var}
\end{figure}

\begin{figure}
    \centering
    \includegraphics[width=0.8\linewidth]{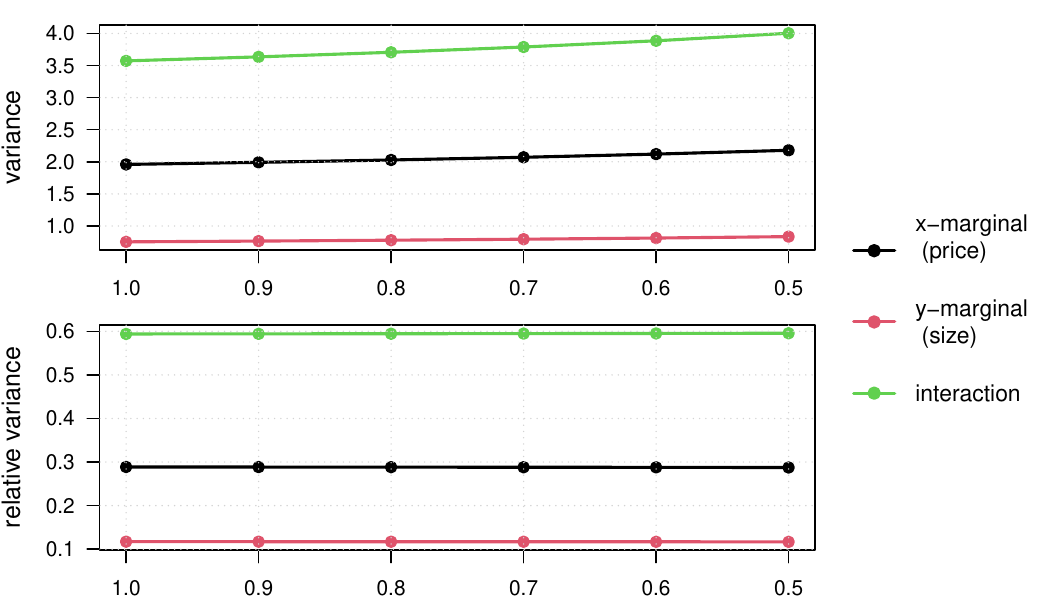}
    \caption{Decomposition of the total variance of the housing density data depending on the imputation constant (multiple of the default value). Top panel: the variance of each part in the orthogonal decomposition; bottom panel: the proportion of variance (relative variance) for each part in the orthogonal decomposition.}
    \label{fig:sim_var_dec}
\end{figure}

Fig.~\ref{fig:sim_scores} shows the FPCA scores and their decomposition for each step of the sensitivity analysis (i.e., multiplication of the mean by a constant in descending order). Similarly, the eigenfunctions can also be explored (see Fig.~\ref{fig:sim_l1} to \ref{fig:sim_l05}). We see that the first principal component and its scores are mostly stable, and capture structural differences in the distributions across US regions. However, the second and third principal component are very close in percent explained variance (9.77~\% and 9.57~\%, cf.\ also the scree plot in Fig.~\ref{fig:h:scree}) and we see that in such  cases, different zero imputation strategies can lead to some differences in results. In particular, some switching seems to take place here between the second and third principal component, starting with the 0.8 multiple. These differences also affect the total scores. Notably, the percentage of variance explained by the first component remains at around 29~\%, and the percentages for the second and third components remain between 9 and 10~\%.

Fig.~\ref{fig:sim_var_dec} displays the decomposition of total variance as the imputation constants decrease. The results are very stable; the proportion of variance for each part in the orthogonal decomposition even does not change at all as the constant in the imputation process is changed.


\subsection{Geological application}
\label{geo}

This application is focused on a Czech geological data set, specifically on concentrations of copper (Cu) and zinc (Zn) in the soil from the RKP data (Register of Contaminated Ares, in Czech Registr Kontaminovan\'ych Ploch), see \cite{Podlesakova1996,Zbiral2004,Polakova2011}.

There are concentrations measured at several locations in 77 Czech districts and we were interested in the bivariate distribution of Cu and Zn concentrations for each of the districts. The original concentrations were measured in $\mathrm{mg} \cdot \mathrm{kg}^{-1}$. Following preprocessing in \cite{skorna24}, the concentrations were log-transformed and the domains were truncated to $[0.588,4.58]$ for Cu and $[1.459,5.663]$ for Zn (similarly as in \cite{MatysGrygar2024}). The final data set contains 49 107 observations and there are from 38 (Brno--město) to 1 570 (Třebíč) observations per district. Due to truncation, 463 observations (0.93~\% of the data) were lost. The bivariate densities of log-concentrations were obtained by a kernel density estimator for each district resulting in discretized densities on a 60$\times$60 grid. This was followed by the imputation of zeros by decreasing values as described in \cite{skorna24}. Once again, all computations were performed in the clr space $\Lz$ after clr transforming all densities.

\begin{figure}
    \centering
    \includegraphics[width=0.8\linewidth]{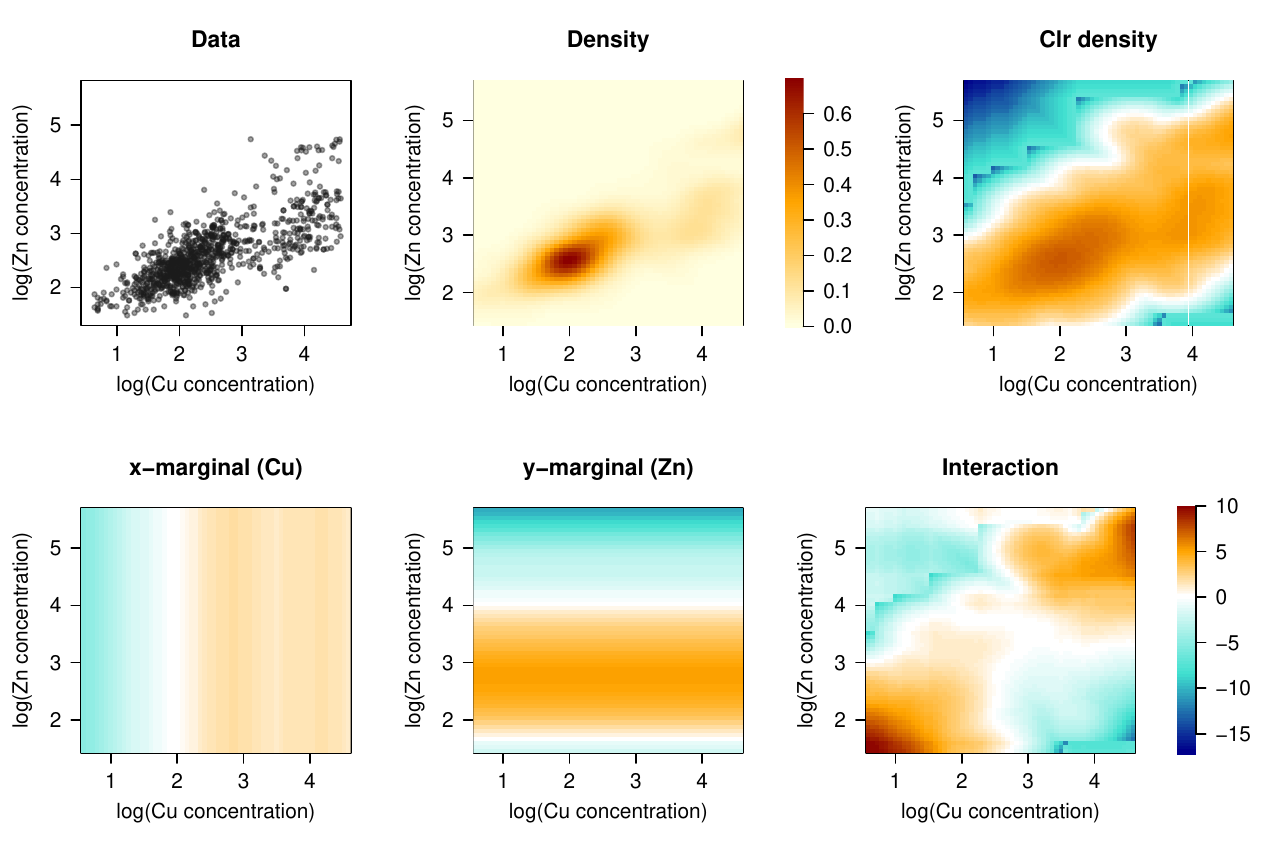}
    \caption{Example data for the Louny district: The first row shows the raw data, the density estimate, and the clr-transformed density. The second row shows the orthogonal decomposition of this density on clr level (the two geometric marginals and the interactive part).}
    \label{fig:g:Louny}
\end{figure}

The results for the Louny district are displayed in Fig.~\ref{fig:g:Louny}. The upper row contains a scatter plot of the raw data, the estimated density and its clr transformation. The bottom row displays the orthogonal parts of the (clr) density, specifically the geometric marginals and the interactive part.

The computational process of FPCA was the same as that for the US housing data (Section~\ref{housing}). The clr densities were first centered and then FPCA for multivariate densities was performed. As a result, the loadings (eigenfunctions) and scores were obtained. The loadings were then decomposed and used to compute decomposed scores.

\begin{figure}[t]
    \centering    
    \includegraphics[width=0.8\linewidth]{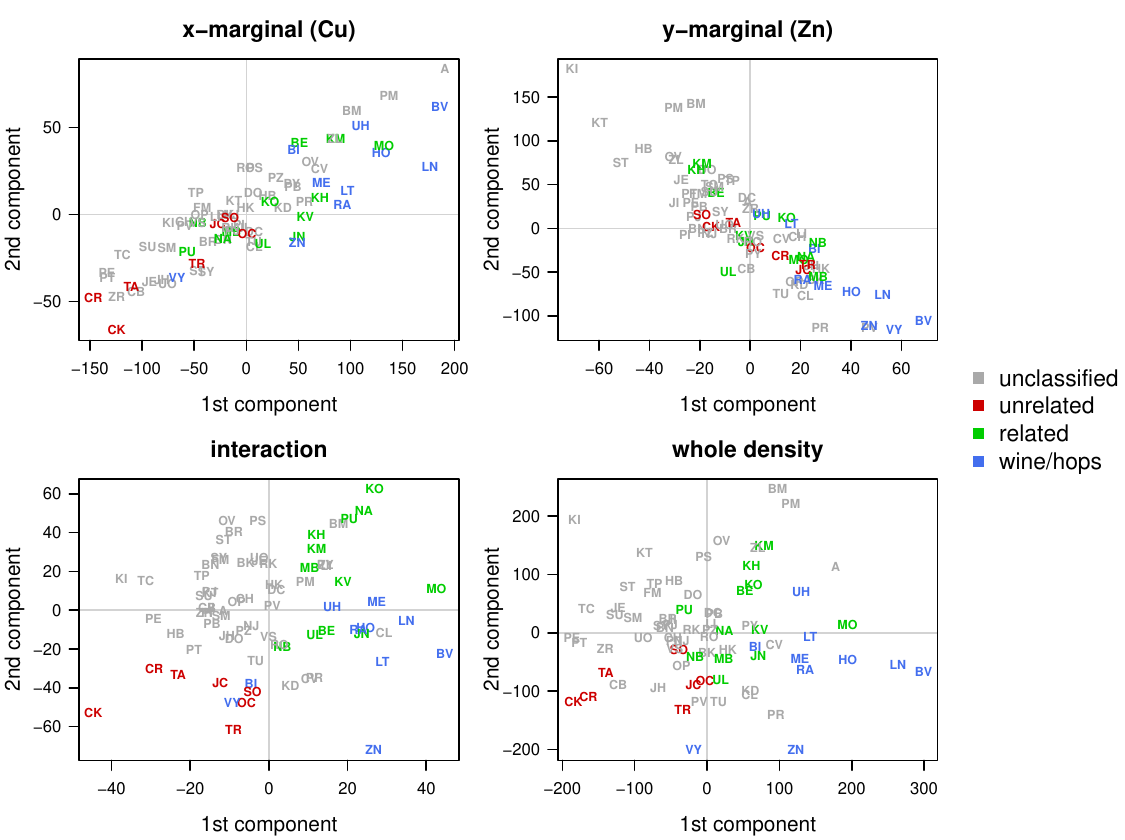}
    \caption{Scores of geological data and their decomposition. The colours corresponds to the selected districts where concentrations of Cu and Zn are unrelated (red), related (green) and where wine or hops are grown (blue). The districts that were not assigned to any of the three groups, are shown in grey.}
    \label{fig:g:scores}
\end{figure}

In Fig.~\ref{fig:g:scores}, the resulting FPCA scores and their decomposition corresponding to the first two principal components are displayed. Unlike the housing data, the colours were assigned only to several chosen districts in this case. The red colour represents the districts where the concentrations of Cu and Zn are not related (these are heterogeneous districts, as follows from the geological background). On the other hand, the districts, where geological relationships between Cu and Zn exist, are green. The blue districts are well known for growing wine or hops; the Cu concentrations are often higher due to the use of Cu pesticides there. There are some structural differences between the three chosen groups of districts, which are mostly visible in the two bottom plots in Fig.~\ref{fig:g:scores} (interaction and total scores), but there are also differences in the marginals.

\begin{figure}[t]
    \centering
    \includegraphics[width=0.85\linewidth]{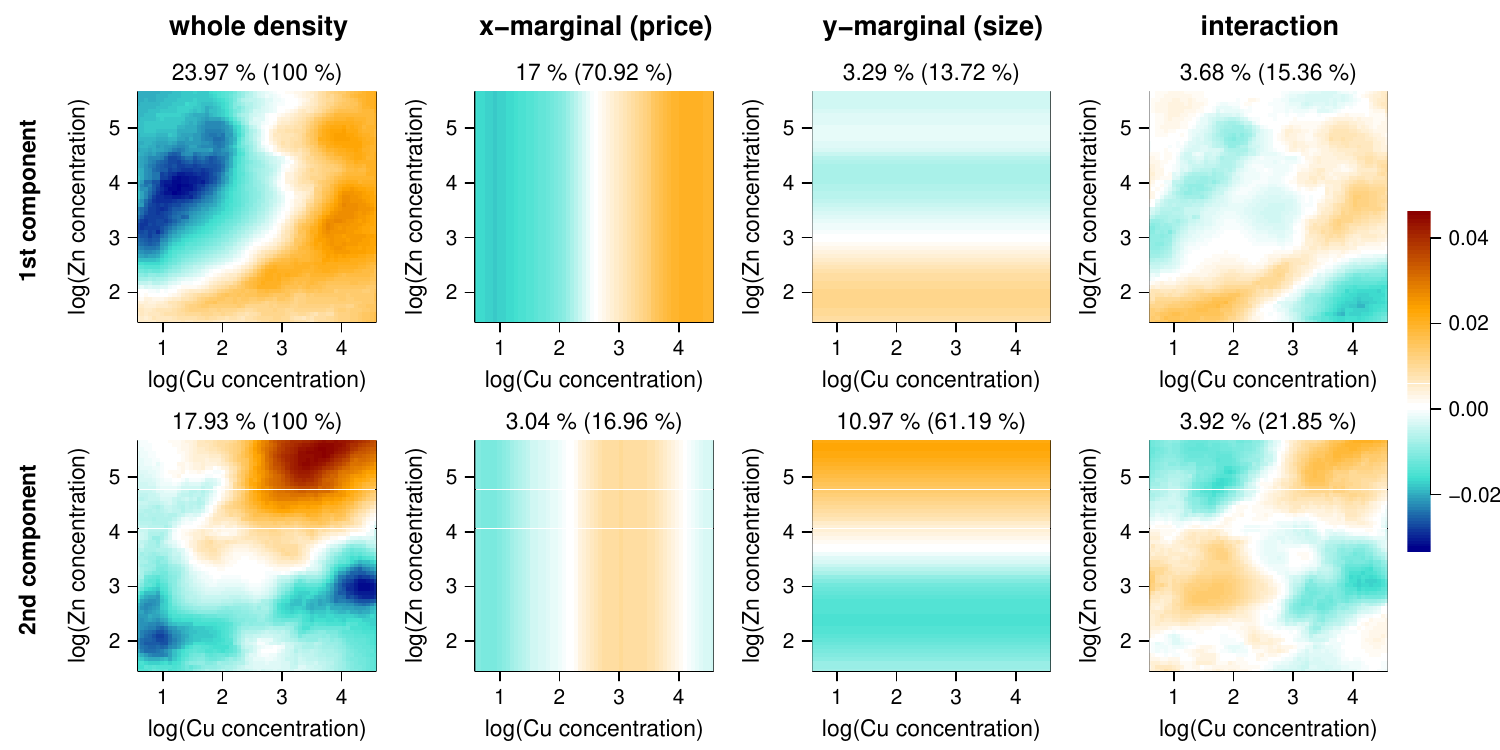}
    \caption{Geological data eigenfunctions and their orthogonal decomposition corresponding to the first two functional principal components on clr level. The percentages show how much variance is explained by each eigenfunction or its part. The percentages in brackets show how much variance of the functional principal component are explained by the particular parts of the eigenfunction.}
    \label{fig:g:loadings}
\end{figure}

Fig.~\ref{fig:g:loadings} shows the loadings and their orthogonal decomposition corresponding to the first two principal components. As in Section~\ref{housing}, these can be used for the interpretation of the scores. As before, marginal loadings are quite easy to interpret because they are connected to high or low concentrations. Firstly, both Cu loadings represent higher Cu concentrations for higher scores, and in connection with the Cu scores, these mostly confirm the fact that there are higher Cu concentrations in wine and hops districts (blue). For instance, Břeclav (BV) is the most typical wine district, while Louny (LN) is the district where the largest quantity of hops is grown. Secondly, the first Zn component is connected to low concentrations, whereas the second one is connected to high concentrations; this means that the Zn scores (top-right panel in Fig.~\ref{fig:g:scores}) show the districts from those with the highest concentrations (top-left corner of the plot) to those with the lowest concentrations (bottom-right corner of the plot). In the wine and hops districts, there are typically lower Zn concentrations. As already mentioned in Section~\ref{housing}, the interaction loadings are more complex, but both indicate positive dependencies between the two elements that are not perfectly linear, so the interaction scores capture the differences between districts with unrelated (red) and related (green) concentrations of Cu and Zn. All of the above results are aggregated into the total scores and loadings, which quite well represent the structural differences between the three selected groups of districts (red, green, blue; grey districts are not assigned to any group).

\begin{figure}
    \centering
    \includegraphics[width=0.6\linewidth]{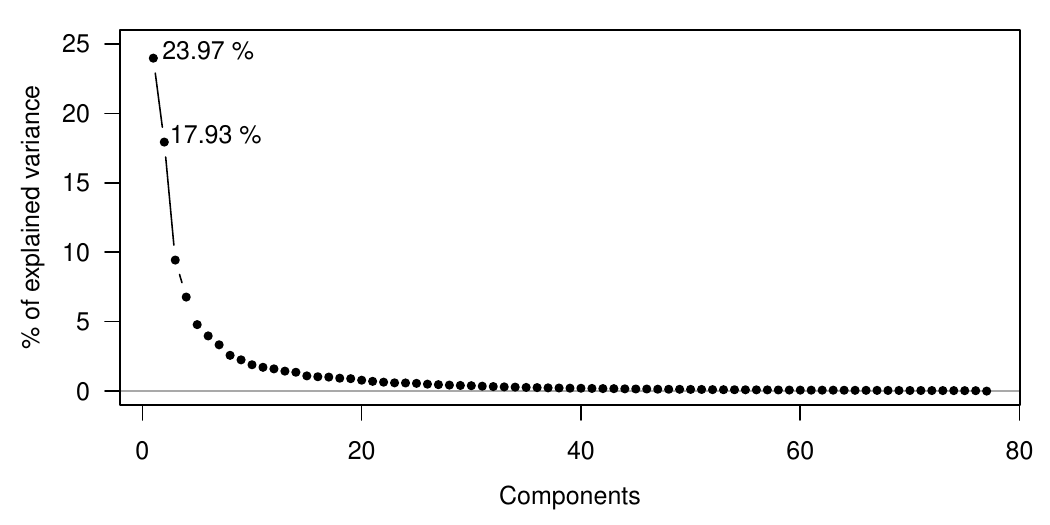}
    \caption{Scree plot for the geological data. The first two principal components together explain 41.9~\% of variance in the data.}
    \label{fig:g:scree}
\end{figure}

Now, if we focus on the explained variability and look at the scree plot (Fig.~\ref{fig:g:scree}), we see that the first principal component explains 23.97~\% of data variance and the second one 17.93~\%. We can also notice that the explained variability drops more slowly than it did for the US housing data (Section~\ref{housing}). Together, the first two components explain 41.9~\% of the variance in the RKP data. More details on the percentages of explained variance are given in Fig.~\ref{fig:g:loadings}.

\begin{figure}[htb]
    \centering
    \includegraphics[width=0.99\linewidth]{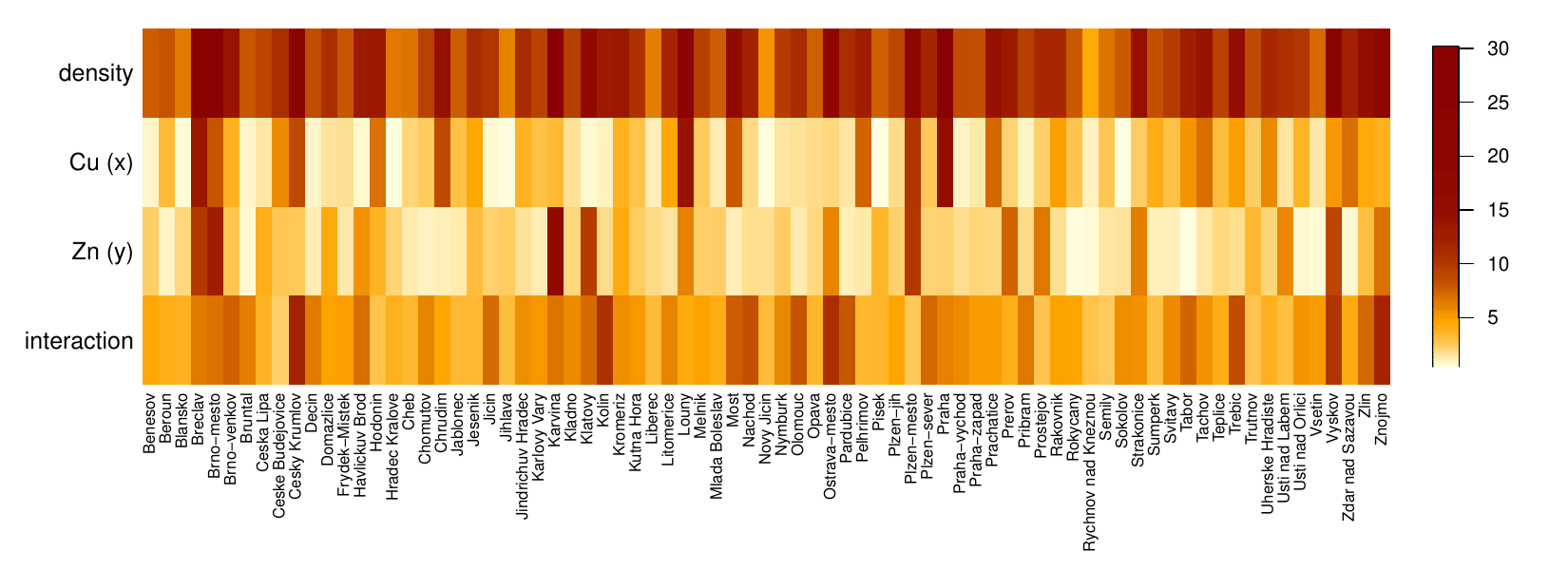}
    \caption{Squared norms of the centered bivariate densities and their decomposition for all the Czech districts.}
    \label{fig:g.norms}
\end{figure}

\begin{figure}[htb]
    \centering
    \includegraphics[width=0.99\linewidth]{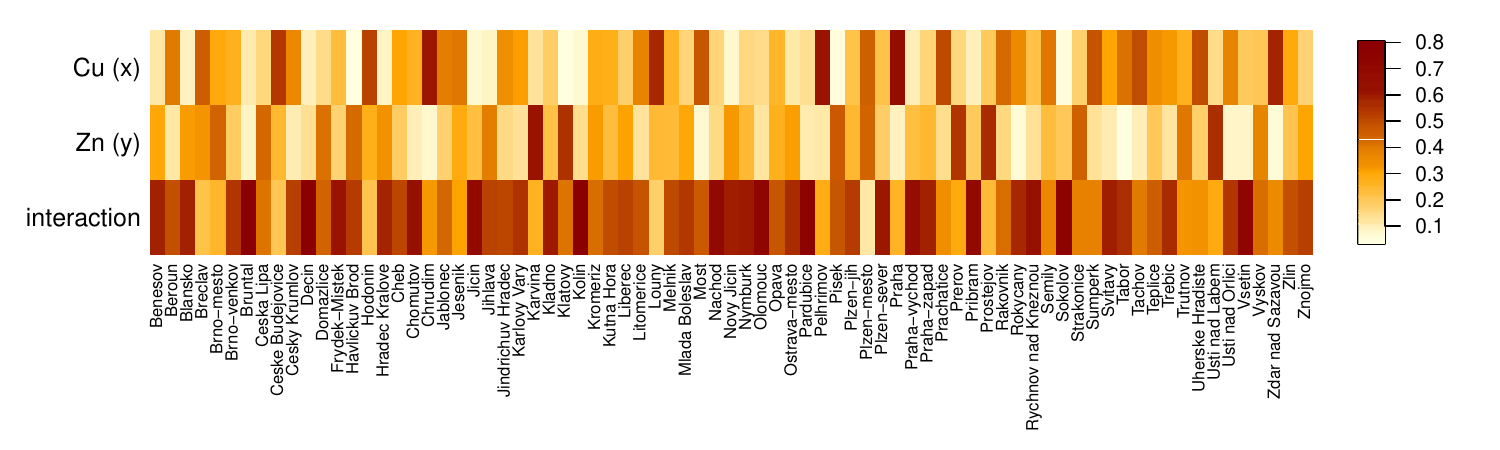}
    \caption{Relative norms of the orthogonal density parts of the geological data.}
    \label{fig:g:rel.norms}
\end{figure}

Similarly as in the previous analysis, we computed squared norms of the centered densities to inspect the decomposition of the variability. In this case, it is not possible to clearly state which part of the decomposition is the most significant due to the differences observed for individual districts. In Fig.~\ref{fig:g.norms}, we can see that while the interactive part seems to be the most relevant for most districts, there are clearly also exceptions where one of the (geometric) marginals is the most important. Overall, the total squared norm appears to be quite large for all districts. For a better comparison of the squared norms of the parts from the decomposition, the relative norms are shown in Fig.~\ref{fig:g:rel.norms}, which confirm the overall strongest importance of the interactive part, but with exceptions.

The variances of the (geometric) marginals are quite similar, the Cu marginals ($x$) contain 29.42~\% of the sample total variance, and the Zn marginals ($y$) contain 25.53~\%. The largest variance is that of the interactive parts, indicating that these parts are the most informative in the density data, as they capture 45.05~\% of the sample total variance.

\begin{figure}
    \centering
    \includegraphics[width=0.9\linewidth]{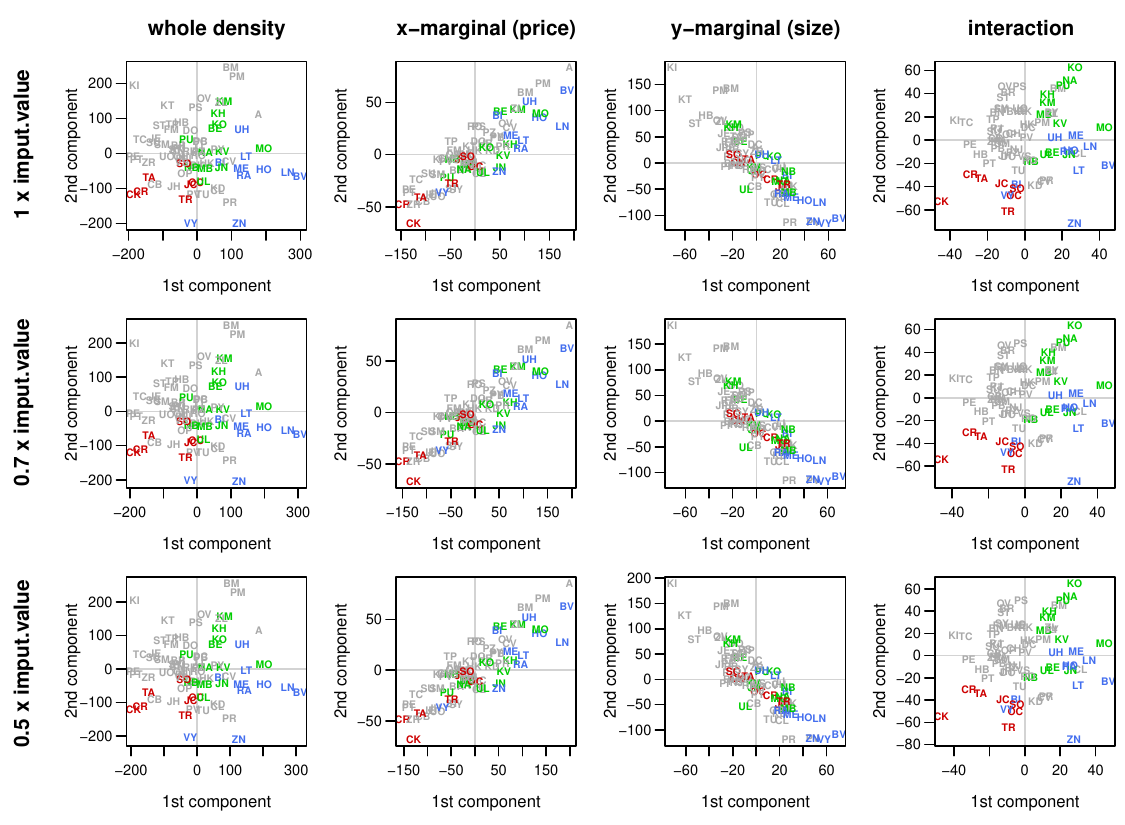}
    \caption{Decomposition of FPCA scores depending on the imputation strategy.}
    \label{fig:sim_g_scores}
\end{figure}

\begin{figure}
    \centering
    \includegraphics[width=0.8\linewidth]{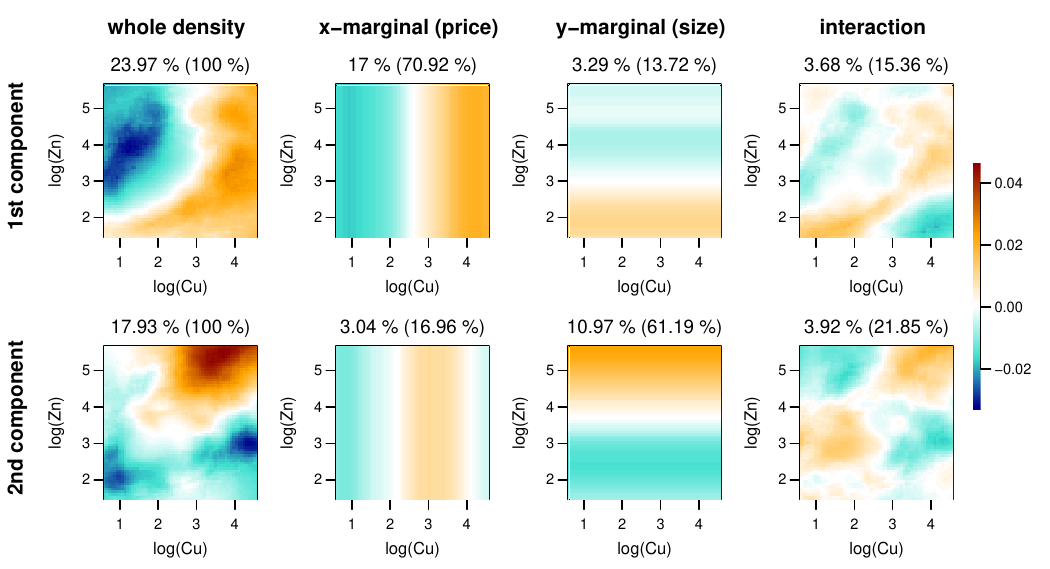}
    \caption{Loadings: $1\times\text{(imputation value)}$}
    \label{fig:sim_g_l1}
\end{figure}

\begin{figure}
    \centering
    \includegraphics[width=0.8\linewidth]{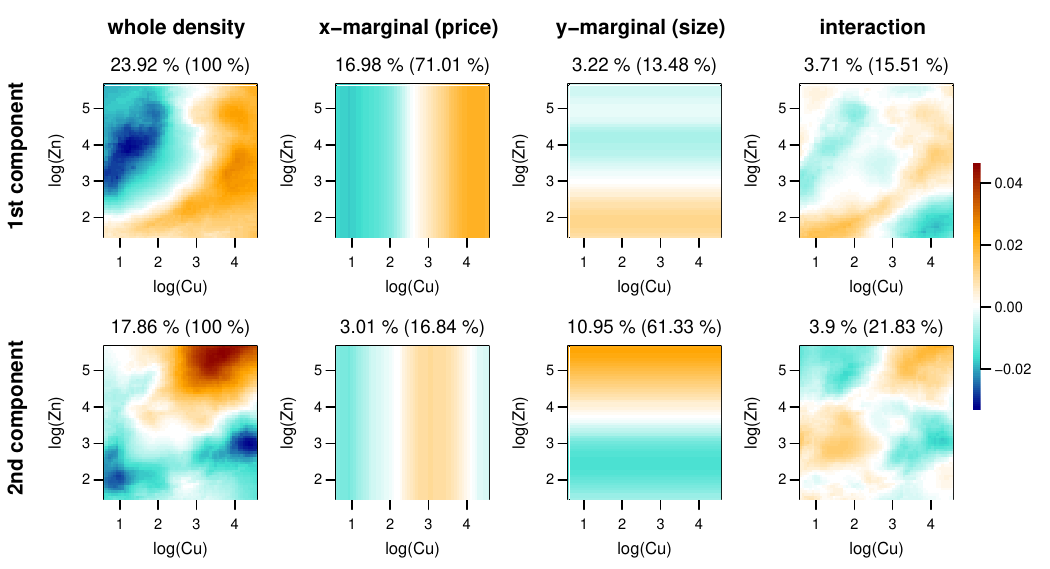}
    \caption{Loadings: $0.7\times\text{(imputation value)}$}
    \label{fig:sim_g_l07}
\end{figure}

\begin{figure}
    \centering
    \includegraphics[width=0.8\linewidth]{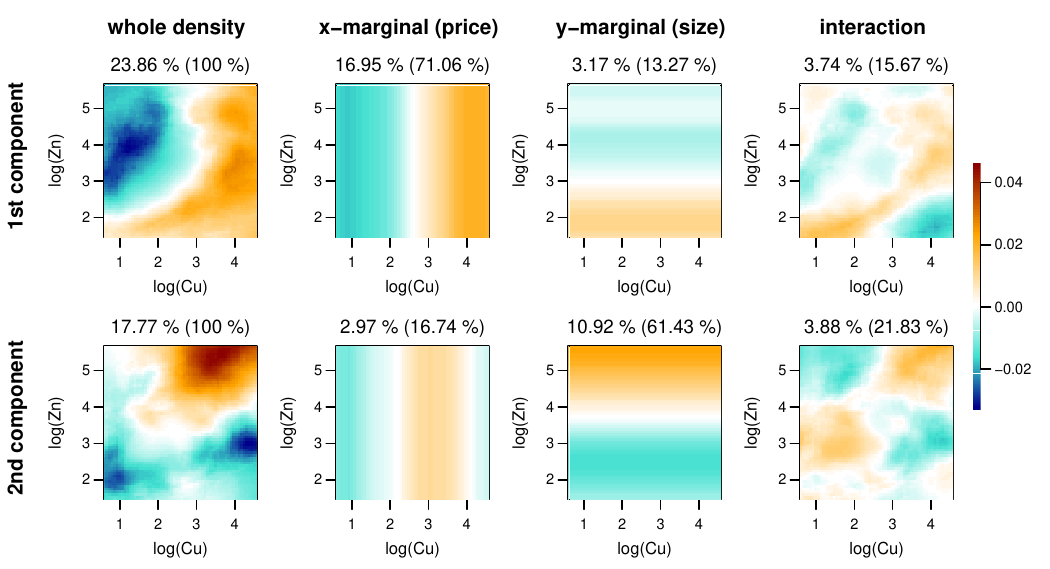}
    \caption{Loadings: $0.5\times\text{(imputation value)}$}
    \label{fig:sim_g_l05}
\end{figure}

\begin{figure}
    \centering
    \includegraphics[width=0.7\linewidth]{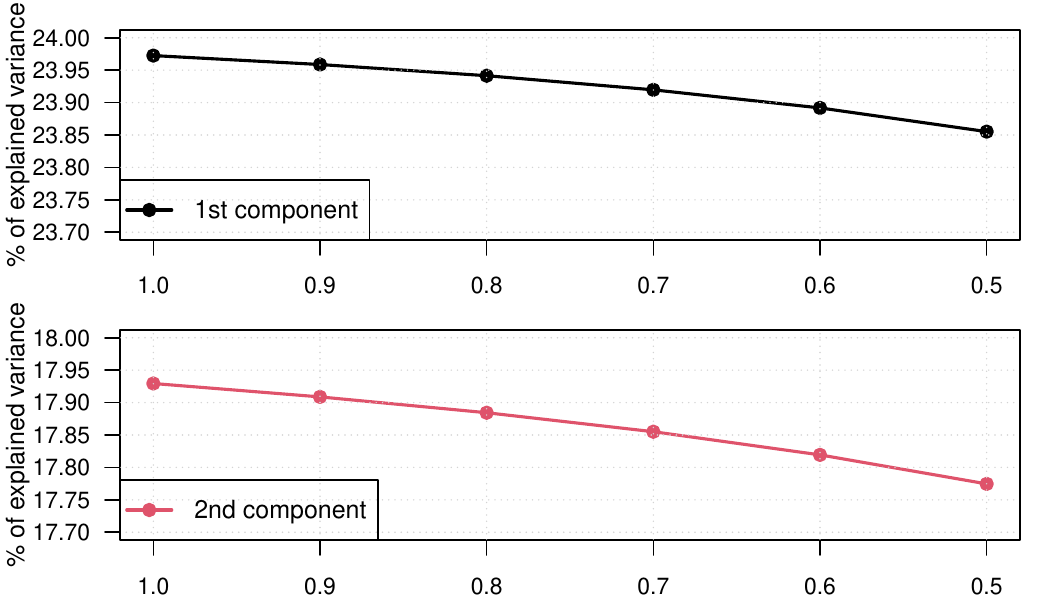}
    \caption{Percentage of explained variance by the first two components depending on the imputation constant (multiple of the default value).}
    \label{fig:sim_g_perc_var}
\end{figure}

\begin{figure}
    \centering
    \includegraphics[width=0.\linewidth]{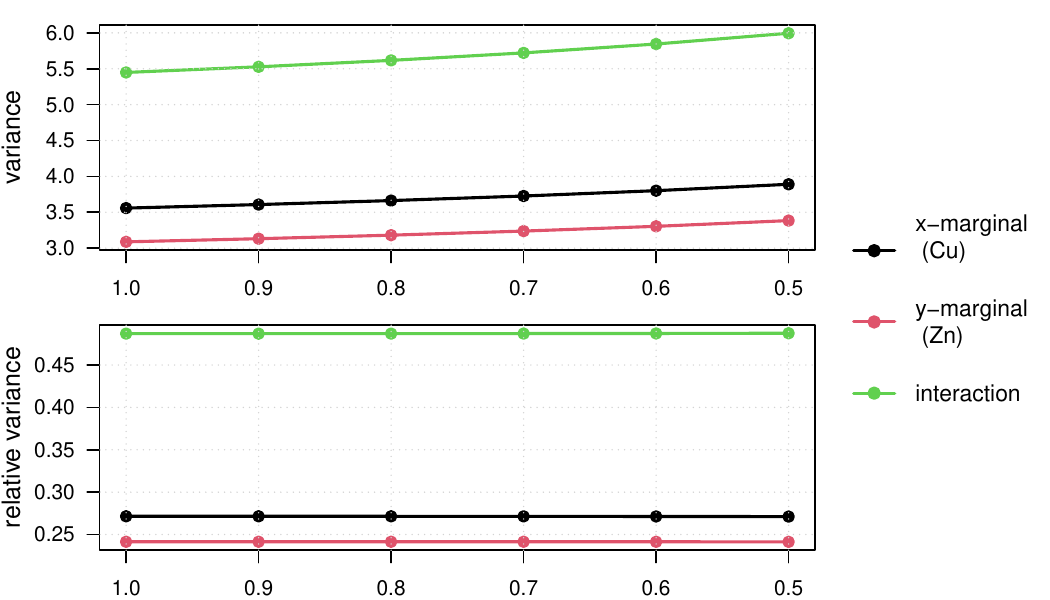}
    \caption{Decomposition of the total variance of the geological density data depending on the imputation constant (multiple of the default value). Top panel: the variance of each part in the orthogonal decomposition; bottom panel: the proportion of variance (relative variance) for each part in the orthogonal decomposition.}
    \label{fig:sim_g_var_dec}
\end{figure}

Also for this geological application, the same sensitivity analysis was performed to demonstrate the robustness of the method against different choices for the zero imputation. Similarly as for the housing data (Section~\ref{housing}), we used multiples of the default imputation values. In this case, the FPCA results are not significantly affected by imputation; scores, eigenfunctions and percentage of explained variance almost do not change (see Figures \ref{fig:sim_g_scores} to \ref{fig:sim_g_perc_var}, for scores and eigenfunctions only the original setting and its 0.7 and 0.5 multiples are shown). An obvious reason is that there are no principal components with close percentages of explained variance  like in the case of the housing data, such that small changes in the data cannot lead to any changes in the order of components. Fig.~\ref{fig:sim_g_var_dec} shows that the decomposition of the total variance is very stable and the proportion of variance for each part in the orthogonal decomposition remains the same across all scenarios.


\subsection{Computational details}
\label{details}

The computational process was basically the same for both examples, in Sections~\ref{housing} and~\ref{geo}. The domains were truncated and variables were log-transformed, then a kernel density estimation resulted in discretized density on a selected grid, i.e., each density was represented by a matrix of density function values. Zero values were imputed, as described in Sections~\ref{housing} and~\ref{geo}, followed by clr transformation, centering, and orthogonal decomposition in the clr space, denoted $\mathbf{F}_n = \mathbf{F}_n^\X + \mathbf{F}_n^\Y + \mathbf{F}_n^\intr$. The resulting matrices were vectorized. The data matrix $\mathbf{F}$ with rows $\mathrm{vec}(\mathbf{F}_n)^\top$ was then constructed with the vectorized (non-decomposed) densities as its rows, so each row corresponds to one of the states or districts. 

The singular value decomposition (SVD) of the matrix $\mathbf{F}$ was performed, i.e., $\mathbf{F} = \mathbf{UDV}^\top$, resulting in the matrix of total scores $\mathbf{S} = \mathbf{UD}$, square roots of variances (diagonal elements of $\mathbf{D}$), and loadings (discretized eigenfunctions) vectorized in the columns of the matrix $\mathbf{V}$. The loadings were transformed into matrices $\mathbf{\Psi}_j$ and orthogonally decomposed, resulting in matrices $\mathbf{\Psi}_j^\X, \mathbf{\Psi}_j^\Y, \mathbf{\Psi}_j^\intr$, then the decomposed scores were computed as inner products of the original centered density parts and the orthogonal parts of the loadings. Practical computations can be simplified when both density parts and loading parts are vectorized, denoted $\mathrm{vec}(\mathbf{F}_n^\X)^\top, \mathrm{vec}(\mathbf{F}_n^\Y)^\top, \mathrm{vec}(\mathbf{F}_n^\intr)^\top, \mathrm{vec}(\mathbf{\Psi}_j^\X)^\top, \\ \mathrm{vec}(\mathbf{\Psi}_j^\Y)^\top, \mathrm{vec}(\mathbf{\Psi}_j^\intr)^\top$, and are the rows of matrices $\mathbf{F}^\X, \mathbf{F}^\Y, \mathbf{F}^\intr, \mathbf{\Psi}^\X, \mathbf{\Psi}^\Y, \mathbf{\Psi}^\intr$. Then each score matrix is obtained as a product of matrix of vectorized original density parts and matrix of vectorized loading parts, i.e., $\mathbf{S}^\X = \mathbf{F}^\X(\mathbf{\Psi}^\X)^\top$, $\mathbf{S}^\Y = \mathbf{F}^\Y(\mathbf{\Psi}^\Y)^\top$, $\mathbf{S}^\intr = \mathbf{F}^\intr(\mathbf{\Psi}^\intr)^\top$.

The percentages of explained variance are then computed using squares of diagonal elements of the matrix $\mathbf{D}$ from SVD. 

The squared norms are the sums of the squared elements of the corresponding matrix of function values, divided by the number of grid points $G$, i.e., $\|\mathbf{F}_n\|^2 = \sum_{i,j} (\mathbf{F}_n)_{ij}^2/G$, and accordingly for $\mathbf{F}_n^\X, \mathbf{F}_n^\Y, \mathbf{F}_n^\intr, \mathbf{\Psi}_j, \mathbf{\Psi}_j^\X, \mathbf{\Psi}_j^\Y$, $\mathbf{\Psi}_j^\intr$.

The variances (the total variance and its decomposition of the centered densities and the loadings) are equal to the sample means of the corresponding squared norms, e.g., $\mathrm{var}(\mathbf{F}) = (1/N) \cdot \sum_n \|\mathbf{F}_n\|^2$, and accordingly for the density parts and the loadings. These are then used to compute the percentages of the total variance.


\section{Conclusions} 
\label{sec:conc}

The introduction of the orthogonal decomposition of multivariate PDFs in \cite{genest23} demonstrated the potential of the Bayes space methodology, as predicted in \cite{petersen22}, and paved the way for further related developments. This paper presents two such developments related to dimension reduction using principal component analysis.

Firstly, it was shown that the variance decomposition of PDFs using their orthogonal decomposition is optimal in the sense of PCA. In other words, the first few densities from the decomposition with the largest variance explain most of the total variance contained in the PDF sample. Exploratory tools such as the scree plot, which is familiar to those who have worked with principal component analysis (see, for example, \cite{johnson07}), can then be used to determine how many elements of the decomposition to take for further processing. Subsequently, the multivariate structure of PDFs can be substantially simplified without significant loss of information.

Secondly, eigenfunctions and scores from FPCA for multivariate densities can be decomposed according to the orthogonal decomposition of PDFs. Moreover, this decomposition of eigenfunctions and scores is unique, as demonstrated by showing that the orthogonal decomposition of FPCA eigenfunctions and scores for the multivariate densities is equivalent to a multivariate FPCA for the vectors of decomposed density components. This enables a structural analysis of eigenfunctions and scores, specifically assessing the role of geometric marginals and various interaction patterns in dimension reduction. For example, it is possible to determine which element(s) of the decomposition are responsible for the outlyingness of a specific observation in the score plot. It can also reveal the source of groupings in a sample of PDFs. However, given the general popularity of (functional) PCA for dimension reduction, the potential for applications is expected to be much broader.

Of course, the Bayes space methodology also has its limitations. The most significant of these is the assumption of square-log integrability of densities and the related amplification of PDF tails. While this follows naturally from the relative scale of densities, it can present challenges in practical implementation, as the estimated densities from raw sampled data are usually affected by the highest approximation error in the tails. The only currently possible solution is to downweight the tails by selecting an appropriate reference measure $\lambda$. Future work could investigate working directly with the raw observations instead of first estimating densities, as done for univariate densities in \cite{steyer23} and \cite{maier2025}. Preprocessing issues must be solved to make the Bayes space approach more accessible to practitioners. On the other hand, the practical inconveniences are amply offset by the deep insight into the density data structure provided by the Bayes space methodology. 


\section*{Acknowledgements}

\noindent AC and KH gratefully acknowledge the support of the Czech Science Foundation, project 25-15447S, and projects IGA\_PrF\_2025\_015, IGA\_PrF\_2026\_018 Mathematical models from the Internal Grant Agency of the Palacký University Olomouc. Funded by the Deutsche Forschungsgemeinschaft (DFG, German Research Foundation) - project number 513634041. We would also like to thank Dr. Tomáš Matys Grygar of the Czech Academy of Sciences for his help in interpreting the geological application.


\bibliographystyle{mybib} 
\bibliography{literature}

\end{document}